\documentclass[a4paper,onecolumn, accepted=2020-09-17]{quantumarticle}
\pdfoutput=1
\usepackage[utf8]{inputenc}
\usepackage[english]{babel}
\usepackage[T1]{fontenc}
\usepackage{amsmath, amsthm, amssymb}
\usepackage{hyperref}
\usepackage[numbers,sort&compress]{natbib}

\usepackage{tikz}
\usepackage{lipsum}

\usepackage{graphicx}
\usepackage{caption, subcaption}
\usepackage{bm}
\usepackage{dsfont}
\usepackage{bbm}
\usetikzlibrary{decorations.markings,intersections}
\usepackage{enumitem}
\hypersetup{linktocpage = true}
\usepackage{tikz-cd}

\usepackage{eucal}
\usepackage{latexsym}
\usepackage{amscd}
\usepackage{overpic}
\usepackage{braket}
\usetikzlibrary{decorations.pathmorphing}
\usepackage{url}

\newcommand{\C}{{\mathbb C}}

\newcommand{\N}{{\mathbb N}}
\newcommand{\Z}{{\mathbb Z}}

\newcommand{\op}{\operatorname}

\newcommand{\End}{\op{End}}

\newcommand{\Hom}{\op{Hom}}

\newcommand{\Tr}{\op{Tr}}

\newcommand{\abs}[1]{{\left\vert{#1}\right\vert}}

\newcommand{\TQO}{\text{TQO}}
\newcommand{\interior}[1]{{#1}^{\circ}}

\newtheorem{theorem}{Theorem}[section]
\newtheorem{proposition}[theorem]{Proposition}

\newtheorem{lemma}[theorem]{Lemma}
\newtheorem{lemma-definition}[theorem]{Lemma-Definition}

\theoremstyle{definition}

\newtheorem{definition}[theorem]{Definition}
\newtheorem{remark}[theorem]{Remark}

\begin{document}

\title{Kitaev's quantum double model as an error correcting code}

\author{Shawn X. Cui}
\email{cui177@purdue.edu}
\affiliation{Departments of Mathematics, Physics and Astronomy, Purdue University, West Lafayette, IN 47907}

\author{Dawei Ding}
\email{dding@stanford.edu}
\author{Xizhi Han}
\email{hanxzh@stanford.edu}
\author{Geoffrey Penington}
\email{geoffp@stanford.edu}
\author{Daniel Ranard}
\email{dranard@stanford.edu}
\author{Brandon C. Rayhaun}
\email{brayhaun@stanford.edu}
\author{Zhou Shangnan}
\email{snzhou@stanford.edu}
\affiliation{Stanford Institute for Theoretical Physics, Stanford University, Stanford, CA 94305}

\maketitle

\begin{abstract}
 Kitaev's quantum double models in 2D provide some of the most commonly studied examples of topological quantum order.  In particular, the ground space is thought to yield a quantum error-correcting code.  We offer an explicit proof that this is the case for arbitrary finite groups. Actually a stronger claim is shown: any two states with zero energy density in some contractible region must have the same reduced state in that region.  Alternatively, the local properties of a gauge-invariant state are fully determined by specifying that its holonomies in the region are trivial.  We contrast this result with the fact that local properties of gauge-invariant states are not generally determined by specifying all of their non-Abelian fluxes --- that is, the Wilson loops of lattice gauge theory do not form a complete commuting set of observables. We also note that the methods developed by P. Naaijkens (PhD thesis, 2012) under a different context can be adapted to provide another proof of the error correcting property of Kitaev's model. Finally, we compute the topological entanglement entropy in Kitaev's model, and show, contrary to previous claims in the literature, that it does not depend on whether the ``log dim R'' term is included in the definition of entanglement entropy.
\end{abstract}

\tableofcontents

\section{Introduction}
Topological phases of matter in two spatial dimensions are gapped quantum liquids that exhibit exotic properties such as stable ground state degeneracy, stable long-range entanglement, existence of quasi-particle excitations, (possibly) non-Abelian exchange statistics, etc. These phases are characterized by a new type of order, topological quantum order ($\TQO$), that is beyond the conventional Landau theory of spontaneous symmetry breaking and local order parameters.\footnote{However, many such phases can be incorporated into the Landau paradigm by generalizing to \emph{non-local} order parameters and spontaneous breaking of \emph{higher form} symmetries \cite{gaiotto2015generalized}.}  An important application of topological phases of matter is in topological quantum computing \cite{kitaev2003fault,freedman2002modular}, where information is encoded in non-local degrees of freedom and processed by manipulating quasi-particle excitations.  

A large class of $\TQO$s can be realized as lattice models in quantum spin systems where the Hamiltonian is given as a sum of pairwise commuting and geometrically local projectors. Examples of such constructions include the Levin-Wen string-net lattice models \cite{levin2005string} and Kitaev's quantum double models \cite{kitaev2003fault}. In \cite{bravyi2010topological, bravyi2011short}, the authors gave a mathematically rigorous proof of gap stability under weak perturbations for quantum spin Hamiltonians satisfying two physically plausible conditions, $\TQO$-1 and $\TQO$-2. Roughly, $\TQO$-1 states that the ground state space is a quantum error correcting code with a macroscopic distance, and $\TQO$-2 means that the local ground state space coincides with the global one. See \cite{bravyi2010topological, bravyi2011short} or \S\ref{subsec:tqo} for a formal definition.


It is widely believed that Kitaev's quantum double models satisfy $\TQO$-1 and $\TQO$-2. This would rigorously justify the inherit fault-tolerance of  Kitaev's models for topological quantum computing.  Previously, it has been shown that Kitaev's model for finite Abelian groups satisfy the $\TQO$ conditions. See \cite{alicki2007statistical, bachmann2017local}, although some of the results in the above mentioned references may be stated in a different form. See also \cite{cha2018complete} for some relevant results. But it is not clear if the methods in those references can be generalized to the case of non-Abelian groups which are significantly more complicated than the Abelian case. In this paper, we provide a proof of the $\TQO$ conditions that work for all finite groups. In fact, we prove a stronger property of Kitaev's model that simultaneously implies $\TQO$-1 and $\TQO$-2. Our result can be informally stated as: 
$$\textit{States with locally zero energy density are locally indistinguishable.}$$
See Theorem \ref{thm:main} for a formal statement. 

Another motivation for the current work is to construct quantum error correcting codes in lattice models. While the toric code and its variant surface code have been studied extensively as error correcting codes, the corresponding study for the case of non-Abelian groups seems insufficient. We have shown in this paper that all non-Abelian Kitaev's models are quantum error correcting codes with macroscopic distance. It will be interesting to see if there are non-Abelian models that outperform the toric code in terms of code properties, such as code distance and threshold. We leave this as a future direction. 

It should be noted that the operator algebra methods developed by P. Naaijkens \cite{naaijkens2012anyons} while studying translation invariant Kitaev's models in the infinite plane (in a different context from the current focus) can also be adapted to our case to give an alternative proof of the main result (Theorem \ref{thm:main}). The current paper, in contrast, considers Kitaev's model on a closed surface and is motivated from the perspective of quantum error correction and topological quantum computing. Our paper showed some of the basic results on Kitaev's model explicitly and we hope that these results and their proof help to reach to a wider community in quantum information science. 

In addition, we show in Section \S\ref{subsec:holonomies} that in Kitaev's model for certain non-Abelian groups, Wilson loops do not form a complete commuting set of observables. That is, there are distinct gauge-invariant states that cannot be distinguished by Wilson loop observables. This is in contrast with toric code where Wilson loop observables completely characterize a gauge-invariant state. For discussions of continuous gauge groups, see \cite{sengupta1994gauge}.

In Section \S\ref{subsec:TEE}, we calculate entanglement entropies in Kitaev's model for the three different definitions of entanglement entropy in gauge theories. In particular, we find that we obtain the same answer for the topological entanglement entropy (TEE) using both the ``algebraic'' and the ``extended Hilbert space'' definitions. This is in contrast with previous claims in the literature \cite{soni2016aspects} that only the extended Hilbert space entropy (which differs from the algebraic entropy by the addition of a ``$\log\mathrm{dim}(R)$'' term) gave the standard answer for the TEE. However, the third definition of entanglement entropy, the distillable entropy, gives a different value for the topological entanglement entropy.

Kitaev's models are actually a special case of Levin-Wen models, which  also conjecturally satisfy the $\TQO$ conditions. Originally, Kitaev's models were only defined for finite groups. However this construction was generalized to finite dimensional Hopf $\C^*$-algebras in \cite{buerschaper2013hierarchy}, and then further generalized to weak Hopf $\C^*$-algebras (or unitary quantum groupoids) in \cite{chang2014kitaev}. On the other hand, the Levin-Wen model takes as input any unitary fusion category. In \cite{chang2014kitaev}, it was proved that the Levin-Wen model associated to a fusion category $\mathcal{C}$ is equivalent to the generalized Kitaev model based on the weak Hopf algebra $H_{\mathcal{C}}$ reconstructed from $\mathcal{C}$ such that $\text{Rep}(H_{\mathcal{C}}) \simeq \mathcal{C}$. Thus, the Levin-Wen models and the generalized Kitaev models are essentially equivalent.

It is an interesting question whether or not our current proof for the case of finite groups can be adapted to the case of Hopf algebras and/or to weak Hopf algebras. For finite groups, there are well-defined notions of local gauge transformations and holonomy which allow us to obtain an explicit characterization of the ground states, though this is not necessary for the proof of our main result. In the general case, such notions are not as clear. We leave these questions for future study.

\section{Background}


In this section, we give a minimal review of a few preliminary notions which are necessary for understanding the proof of our main theorem. We begin by discussing generalities related to error correcting codes, topological quantum order, and the relationship between them, and then describe the particular models which we will be studying. 

\subsection{Error correcting codes}
We provide a very brief introduction to quantum error correcting codes (QECCs), mainly to set up the conventions that will be used later. For a detailed account of the theory of QECCs, we recommend \cite{nielsonchuang}. 

To protect quantum information against noise, a common strategy is to embed states $|\psi\rangle$ which contain information into a subspace $\mathcal{C}$, called the code subspace, of a larger Hilbert space $\mathcal{H}$. Quantum processing of the state is then modeled as a noisy quantum channel $\mathcal{E}$, which is a completely positive, trace preserving map on the density matrices living in $\mathcal{H}$. It is possible to successfully retrieve the information contained in $|\psi\rangle$ if there is another recovery quantum channel $\mathcal{R}$ such that
\begin{equation}
    (\mathcal{R} \circ \mathcal{E})(|\psi\rangle \langle \psi|) = |\psi\rangle \langle \psi | \text{ for any } |\psi\rangle \in \mathcal{C}.
\end{equation}
The recovery only needs to be perfect for states in the code subspace, and the larger Hilbert space acts as a resource of redundancy that makes the recovery possible. 

Any quantum channel $\mathcal{E}$ can be written as the composition of an isometry $V: \mathcal{H} \to \mathcal{H} \otimes \mathcal{H}_E$ together with a partial trace over the `ancilla' degrees of freedom $\mathcal{H}_E$ as
\begin{align}
\mathcal{E} (\rho) = \Tr_E ( V \rho V^\dagger).
\end{align}
This representation is unique up to isomorphisms of $\mathcal{H}_E$. If we choose some computational basis $\{\ket{i} \in \mathcal{H}_E\}$ for the ancilla system and make the partial trace explicit, we obtain an `operator-sum representation' (or Kraus decomposition) for the quantum channel $\mathcal{E}$, given by
\begin{equation}
    \mathcal{E}(\rho) = \sum_i E_i \rho E_i^\dagger,
\end{equation}
where the \emph{operation elements} $E_i \in \End{\mathcal(\mathcal{H})}$ are defined by $E_i = \bra{i} V$. The isometry condition $V^\dagger V = I$ becomes $$\sum_i E_i^\dagger E_i = I.$$ 
For a noisy quantum channel, the $E_i$ can be thought of as the operators that create errors. A general theorem concerning the existence of recovery channels can be found in 10.3 of \cite{nielsonchuang}, which we reproduce below:
\begin{theorem}\label{kl_theorem}
Let $P \in \End(\mathcal{H})$ be the projection onto the code subspace $\mathcal{C}$, and $\mathcal{E}$ a quantum channel with operation elements $\{E_i\}$. A necessary and sufficient condition for the existence of a recovery channel $\mathcal{R}$ correcting $\mathcal{E}$ on $\mathcal{C}$ is that
\begin{equation}
    P E_i^\dagger E_j P = \alpha_{i j} P, \label{eq:error_correcting}
\end{equation}
for some Hermitian matrix $\alpha$ of complex numbers. 
\end{theorem}
In later sections, we will prove that the ground state space of Kitaev's quantum double model is a quantum error correcting code by showing that \eqref{eq:error_correcting} holds. 
\subsection{Topological quantum order}\label{subsec:tqo}
We now review the definition of topological quantum order ($\TQO$) introduced in \cite{bravyi2010topological}.
Let $\Lambda = (V(\Lambda),E(\Lambda),F(\Lambda))$ be an $L\times L$ lattice with periodic boundary conditions. The requirement that the lattice has periodic boundary is purely for the sake of simplicity. In general, one can take any lattice of linear size $L$ that lives on a surface of arbitrary genus.
In Kitaev's quantum double model, the qudits are conventionally defined to live on the edges of $\Lambda$ instead of the vertices; for simplicity we use the same convention here.\footnote{The choice of whether the qudits live on the edges or vertices of the lattice is arbitrary and makes no difference to the definition.} We therefore associate to each edge $e \in E(\Lambda)$ a qudit $\mathcal{H}_e=\C^d,$ and take the total Hilbert space to be $\mathcal{H} = \bigotimes_{e \in E(\Lambda)} \mathcal{H}_e$. We consider Hamiltonians of the form
\begin{equation}
    H = \sum\limits_{v \in V(\Lambda)} (1-P_{v}) + \sum\limits_{f \in F(\Lambda)} (1-P_f),
\end{equation}
where $P_v$ is a projector that acts non-trivially only on edges which meet the vertex $v$, and $P_f$ is a projector that acts non-trivially only on the boundary edges of the plaquette $f$. We further demand that the $P_v$'s and the $P_f$'s mutually commute and that the Hamiltonian be frustration free, i.e.\ that the ground states of $H$ are stabilized by each $P_v$ and each $P_f$:
\begin{equation}
    V_{\mathrm{g.s.}} = \{\ket{\psi} \in \mathcal{H}: P_v \ket{\psi} = \ket{\psi} \text{ and } P_f \ket{\psi} = \ket{\psi}, \ \forall \ v \in V(\Lambda), f \in F(\Lambda) \}.
\end{equation}
Denote the projection onto $V_{\mathrm{g.s.}}$ by $P$, which can be written as
\begin{equation}
    P = \prod\limits_{v \in V(\Lambda)}P_v \prod\limits_{f \in F(\Lambda)} P_f.
\end{equation}
Let $A$ be a sublattice of $\Lambda$ of size $\ell \times \ell$, denote by $\interior{V(A)}$ the subset of $V(A)$ that are in the interior of $A$ (which is of size $(\ell - 2) \times (\ell - 2)$), and define 
\begin{equation}
    P_{A} = \prod\limits_{v \in \interior{V(A)}}P_v \prod\limits_{f \in F(A)} P_f.
\end{equation}
We can now state the definition of TQO that we will use. 

\begin{definition}[Topological Quantum Order \cite{bravyi2010topological}]

A Hamiltonian which is frustration-free is said to have topological quantum order (TQO) if there is a constant $\alpha > 0$ such that for any $\ell \times \ell$ sublattice $A$ with $\ell \leq L^{\alpha}$, the following hold.
\begin{itemize}
    \item $\TQO$-1:  For any operator $O$ acting on $A$, 
\begin{equation}
    P O P = c_O P,
\end{equation}
where $c_O$ is some complex number.
\item $\TQO$-2: If $B$ is the smallest square lattice whose interior properly contains $A$,\footnote{So $B$ has size $(\ell +2) \times (\ell +2)$.} then $\Tr_{\bar{A}}(P)$ and $\Tr_{\bar{A}}(P_{B})$ have the same kernel, where $\bar{A}$ is the complement of $A$ in $\Lambda$.
\end{itemize}
\end{definition}

TQO-1 heuristically corresponds to the statement that a sufficiently local operator cannot be used to distinguish between two orthogonal ground states because they differ only in their global, ``topological'' properties. Furthermore, ground state denegeracy is ``topologically protected'' in systems satisfying TQO-1 in the sense that perturbations by local operators can induce energy level splitting only non-perturbatively, or at some large order in perturbation theory which increases with the size of the lattice. It is straightforward to show that $\TQO$-1 is equivalent to the condition that all normalized ground states $\ket{\psi} \in V_{\mathrm{g.s.}}$ have the same reduced density matrix on $A$.

TQO-2 is the statement that the local ground state spaces and the global one should agree. We emphasize that TQO-2 can be violated at regions with non-trivial topologies, which is why one restricts to square lattices. 

\begin{remark}
For our purposes, TQO-1 and QECC are morally interchangeable. Indeed, if $H$ is any Hamiltonian\footnote{For defining TQO-1, we do not need that $H$ is frustration free.} with $P$ the projection onto the ground space $V_{\mathrm{g.s.}}$, then the following are equivalent.
\begin{enumerate}
    \item The Hamiltonian $H$ has TQO-1.
    \item The Hamiltonian $H$ provides a QECC with code subspace $V_{\mathrm{g.s.}}$. There exists an $\alpha>0$ such that the code can correct any error $\rho \mapsto \sum_i E_i\rho E_i^\dag$ for which every combination $E_i^\dag E_j$ is supported on an $\ell\times \ell$ sublattice $A$ with $\ell\leq L^\alpha$.  
\end{enumerate}
\end{remark}

In \S\ref{subsec:indistinguishable}, we will prove a theorem for Kitaev's finite group models which simultaneously implies TQO-1 and TQO-2, and so by the above remark also implies that the model furnishes a QECC.

\subsection{Kitaev's finite group lattice model}
We now turn to Kitaev's finite group lattice models \cite{kitaev2003fault}, which we will see instantiate the concepts of the previous sections. Let $G$ be a finite group, $\Sigma$ be an oriented 2D surface with no boundary, and $\Lambda = (V, E, F)$\footnote{We abbreviate $V\equiv V(\Lambda)$, $E\equiv E(\Lambda)$, and $F\equiv F(\Lambda)$.} be an arbitrary oriented lattice on $\Sigma$, where $V$, $E$, and $F$ are the sets of vertices, oriented edges, and plaquettes of the lattice, respectively. Then, for every $e \in E$, set $\mathcal{H}_e = \mathbb{C}[G]$ the group algebra of $G$, i.e.\ $\mathcal{H}_e$ is spanned by the basis $\left\{ \ket{g} : g \in G \right\}$. The overall Hilbert space is given by $\mathcal{H} \equiv \bigotimes_{e \in E} \mathcal{H}_e$. A natural basis for this Hilbert space consists of tensor products of the form $\ket{g}\equiv \bigotimes_{e\in E}|g_e\rangle$; we refer to this as the \textit{group basis}.

We define the \emph{sites} of $\Lambda$ to be the set of pairs $s=(v,p)\in V\times F$ such that $p$ is adjacent to $v$. Given a site $s = (v,p)$ and two elements $g, h$ in $G$, we define two sets of operators: gauge transformations $A_v(g)$ and magnetic operators $B_{(v,p)}(h)$. Their action is most readily seen in the group basis. For example, $A_v(g)$ acts on the edges which touch $v$ by multiplication by $g$ on the left, or multiplication by $g^{-1}$ on the right, depending on whether the edge is oriented away from or towards $v$. The magnetic operator $B_{(v,p)}(h)$ computes the product of the group elements sitting on the edges of $p$, and compares it to $h$, annihilating the state if there is a discrepancy, while stabilizing it if the group elements agree. The prescription for computing the product is to start at $v$ and move around $p$ counter-clockwise, inverting the group element associated to an edge if that edge is oriented opposite relative to the direction of travel. For example,
\begin{align}
\begin{split}
  A_v(g) \Bigg\vert
  \begin{tikzpicture}[baseline={([yshift=-.5ex]current bounding box.center)},vertex/.style={anchor=base,circle,fill=black!25,minimum size=18pt,inner sep=2pt}]
    \draw[fill] (0,0) circle [radius=0.025] node[above right] {$v$};
    \draw (-1,0) -- (1,0);
    \draw (0,-1) -- (0,1); 
    \draw[->] (0,0) -- (-0.5,0) node[above] { $g_1$};
    \draw[->] (0,0) -- (0,-0.5) node[left] { $g_2$};    
    \draw[->] (0,0) -- (0.5,0) node[below] { $g_3$};
    \draw[->] (0,1) -- (0,0.5) node[above right] { $g_4$};
  \end{tikzpicture}
\Bigg\rangle & \equiv 
  \Bigg\vert
  \begin{tikzpicture}[baseline={([yshift=-.5ex]current bounding box.center)},vertex/.style={anchor=base,circle,fill=black!25,minimum size=18pt,inner sep=2pt}]
    \draw[fill] (0,0) circle [radius=0.025] node[above right] {$v$};
    \draw (-1,0) -- (1,0);
    \draw (0,-1) -- (0,1);
    \draw[->] (0,0) -- (-0.5,0) node[above] {$gg_1$};
    \draw[->] (0,0) -- (0,-0.5) node[left] {$gg_2$};    
    \draw[->] (0,0) -- (0.5,0) node[below] {$gg_3$};
    \draw[->] (0,1) -- (0,0.5) node[above right] {$g_4g^{-1}$};
  \end{tikzpicture}
  \Bigg\rangle \\ 
  B_{(v,p)}(h) \Bigg \vert
  \begin{tikzpicture}[baseline={([yshift=-.5ex]current bounding box.center)},vertex/.style={anchor=base,circle,fill=black!25,minimum size=18pt,inner sep=2pt}]
    \draw[fill] (0,0) circle [radius=0.025] node[below left] {$v$};
    \draw (0,0) -- (1,0) -- (1,1) -- (0,1) -- (0,0); 
    \draw[->] (0,0)--(0.5,0) node[below] {$h_1$};
    \draw[->] (1,1)--(1,0.5) node[right] {$h_2$};
    \draw[->] (1,1)--(0.5,1) node[above] {$h_3$};
    \draw[->] (0,1)--(0,0.5) node[left] {$h_4$};
    \node at (0.5,0.5) {$p$};
  \end{tikzpicture}
  \Bigg \rangle & \equiv
  \delta_{h, h_1 h_2^{-1} h_3 h_4} 
  \Bigg \vert
  \begin{tikzpicture}[baseline={([yshift=-.5ex]current bounding box.center)},vertex/.style={anchor=base,circle,fill=black!25,minimum size=18pt,inner sep=2pt}]
    \draw[fill] (0,0) circle [radius=0.025] node[below left] {$v$};
    \draw (0,0) -- (1,0) -- (1,1) -- (0,1) -- (0,0); 
    \draw[->] (0,0)--(0.5,0) node[below] {$h_1$};
    \draw[->] (1,1)--(1,0.5) node[right] {$h_2$};
    \draw[->] (1,1)--(0.5,1) node[above] {$h_3$};
    \draw[->] (0,1)--(0,0.5) node[left] {$h_4$};
    \node at (0.5,0.5) {$p$};
  \end{tikzpicture}
  \Bigg \rangle
  \end{split}
\end{align}
where $\delta_{g,h}$ is the Kronecker delta symbol. Some basic facts follow:
\begin{align}\label{operator_properties}
\begin{split}
  A_v(g) A_v(h) & = A_v(gh)\\
  B_{(v,p)}(g) B_{(v,p)}(h) & = \delta_{g, h} B_{(v,p)}(h)\\
  A_v(g) B_{(v,p)}(h) & = B_{(v,p)}(g h g^{-1}) A_v(g).
  \end{split}
\end{align}
We can now define the vertex and plaquette operators as
\begin{align}
\begin{split}
  A_v & \equiv \frac{1}{\abs{G}} \sum_{g \in G}^{} A_v(g) \\
  B_p & \equiv B_{(v,p)}(\mathds{1}),
  \end{split}
\end{align}
where $v$ is any vertex adjacent to $p$ and $\mathds{1} \in G$ is the identity element.\footnote{When $g$ is the identity element, the definition of $B_{(v,p)}(g)$ depends only on the plaquette in $s = (v,p)$, not the vertex.} It is easily verified that for all $v \in V$, $p \in F$, $A_v$ and $B_p$ are commuting projectors. The Hamiltonian of this system is defined in terms of these projectors:
\begin{align}
  H = \sum_{v \in V}^{} (1 - A_v) + \sum_{p \in F}^{} (1 - B_p).
\end{align}
This Hamiltonian is frustration-free and the ground space is simply given by
\begin{align}
  V_{\text{g.s.}} \equiv \left\{ \ket{\psi} \in \mathcal{H} :  A_v \ket{\psi} = B_p \ket{\psi} = \ket{\psi}, \  \forall v \in V, \  p \in F \right\}.
\end{align}
In gauge-theoretic language, where we think of a state as specifying the field configuration of a $G$ vector potential, the condition that $A_v|\psi\rangle = |\psi\rangle$ means that $|\psi\rangle$ is gauge invariant, while $B_{(v,p)}(h)|\psi\rangle = |\psi\rangle$ means that the connection is flat. Now, due to the identities
\begin{align}
\begin{split}
  A_v(g) A_v & = A_v\\
  B_{(v,p)}(h) B_p & = \delta_{h,\mathds{1}} B_p,
  \end{split}
\end{align}
the action of the $A_v(g)$ and $B_{(v,p)}(h)$ operators on the ground space is simply 
\begin{align}
\begin{split}
  A_v(g) \ket{\psi} & = \ket{\psi}\\
  B_{(v,p)}(h) \ket{\psi} & = \delta_{h,\mathds{1}} \ket{\psi}
  \end{split}
\end{align}
for all $\ket{\psi} \in V_\text{g.s.}$. In Section \ref{subsec:gs_space}, we show that the dimension of $V_\text{g.s.}$ is the number of orbits of $\Hom(\pi_1(\Sigma), G)$ under the action of $G$ by conjugation, where $\pi_1(\Sigma)$ is the fundamental group of $\Sigma$.

We recall that the toric code is the ground space of the above Hamiltonian for $\Sigma = T^2$ the two-torus, $\Lambda$ an $L \times L$ periodic square lattice, and $G = \Z_2$. In this case, the orientations of the edges in $E$ does not matter and we can identify $\C[G]$ with a qubit, with the two elements $0,1$ of $\Z_2$ corresponding to $\ket{0}, \ket{1}$ of the computational basis. It is easy to check that 
\begin{align}
  A_v  = \frac{1+X_v}{2}, \, B_p  = \frac{1+Z_p}{2},
\end{align}
where $X_v$ is the tensor product of Pauli $X$ operators on all the Hilbert spaces in the edges incident to $v$, and $Z_p$ is the tensor product of Pauli $Z$ operators on the edges on the boundary of $p$. The ground space is spanned by states corresponding to homology classes of loops on a torus. This is a consequence of the explicit characterization of the ground space  corresponding to any finite group $G$ in the next section.

\subsection{Ground state space of Kitaev's model}
\label{subsec:gs_space}
In this subsection, we discuss some properties of the ground state and count the ground state degeneracy. This result is stated in the original paper without a proof \cite{kitaev2003fault} and is known to experts in relevant areas. See also \cite{bachmann2017local} for the case of cyclic groups. However, we did not find a reference that addresses the general case explicitly. Therefore, we think it is beneficial to the readers to provide a detailed and elementary derivation. We follow the notations from the previous subsection.
There is an action of $G$ on $\Hom(\pi_1(\Sigma), G)$ by conjugation: for $g \in G$ and $ \phi \in \Hom(\pi_1(\Sigma), G)$, we set $(g\cdot \phi)(.) \equiv g \phi(.) g^{-1}$.

\begin{theorem}
The dimension of $V_{\mathrm{g.s.}} (\Sigma)$ is equal  to the number of orbits in $\Hom(\pi_1(\Sigma), G)$ under the $G$-action.
\end{theorem}

\begin{proof}
A basis element $|g\rangle = \bigotimes_{e\in E}|g_e\rangle$ of the total Hilbert space is an assignment of a group element $g_{e}$ to each edge $e\in E$. Let $\gamma$ be any oriented path in the lattice, which can be thought of as a sequence of connected edges. The group element obtained by multiplying the group elements along the path is denoted by $g_{\gamma}$. If one edge is oriented opposite to the path, then we multiply the inverse of the group element of that edge. 

The constraint $B_p |g\rangle = |g\rangle$ is equivalent to the condition that $g_{\partial p} = \mathds{1}$, where $\partial p$ is the boundary of $p$ oriented counterclockwise, thought of as a path.\footnote{For testing whether or not $g_{\partial p}=\mathds{1}$, it does not matter which vertex we think of $\partial p$ as starting at.} Hence, the subspace fixed by all the $B_p$'s is spanned by the following set:
\begin{align}
\begin{split}
    S &= \{ |g\rangle :\  g_{\partial p} = \mathds{1}, \ \forall p\in F\} \\
    &= \{|g\rangle: \  g_{\gamma} = \mathds{1}, \ \text{for any \emph{contractible}, closed } \gamma\}
    \end{split}
\end{align}

For any $h\in G$, we call the operator $A_v(h)$ a gauge transformation at the vertex $v$. For two basis elements $|g\rangle$, $|g'\rangle \in S$, we call $|g\rangle$ and $|g'\rangle$ gauge equivalent if $|g'\rangle$ can be obtained from $|g\rangle$ by applying some gauge transformations at several vertices, denoted by $|g\rangle \sim |g'\rangle$. Gauge equivalence defines an equivalence relation on $S$. We denote the set of equivalence classes by $[S]$.

For each $[g] \in [S]$, define

\begin{equation}
    |[g]\rangle := \sum_{|g\rangle \sim |g'\rangle} |g'\rangle
\end{equation}
Since 
\begin{equation}
    A_v(h) |[g]\rangle = \sum_{|g\rangle \sim |g'\rangle} A_v(h) |g'\rangle = \sum_{|g\rangle \sim |g''\rangle} |g''\rangle = |[g]\rangle,
\end{equation}
this implies that $|[g]\rangle$ is stabilized by $A_v$,
\begin{equation}
    A_v |[g]\rangle = \frac{1}{|G|} \sum_{h \in G} A_v(h) |[g]\rangle = |[g]\rangle.
\end{equation}
We conclude that $|[g]\rangle \in V_{\mathrm{g.s.}}(\Sigma)$. It is direct to check that $\{ |[g]\rangle: \ [g] \in [S] \}$ forms a basis of $V_{\mathrm{g.s.}}(\Sigma)$.

We now build a correspondence between $[S]$ and orbits in $\Hom(\pi_1(\Sigma), G)$. Choose any vertex $v_0$ as a base point of $\Lambda$ and choose a maximal spanning tree $T$ containing $v_0$. By definition, a maximal spanning tree is a maximal subgraph of the lattice $\Lambda$ that does not contain any loops. Hence, any maximal spanning tree contains exactly $m := |V| - 1$ edges.

We define a map 
\begin{equation}
    \Phi: S \longrightarrow \Hom(\pi_1(\Sigma), G)
\end{equation}
as follows. Let $\gamma$ be any closed path starting and ending at $v_0$. For any $|g\rangle \in S$, define $\Phi(|g\rangle) ([\gamma]) := g_{\gamma} $. Namely, $\Phi(|g\rangle) $ maps a closed path $\gamma$ to the product of the group elements on it. The fact that $g_{\gamma_0} = \mathds{1}$ for any contractible loop $\gamma_0$ implies that $\Phi(|g\rangle) ([\gamma]) $ only depends on the homotopy class of $\gamma$. Hence, $\Phi(|g\rangle) $ is a well defined map from $\pi_1(\Sigma, v_0)$ to $G$.\footnote{As is standard in algebraic topology, the choice of basepoint is immaterial in defining the fundamental group up to isomorphism, so we suppress it from the notation from now on.} It is clear that it is also a group homomorphism, so
\begin{equation}
    \Phi(|g\rangle) \in \Hom(\pi_1(\Sigma), G)
\end{equation}
Now we show that $\Phi$ is onto and in fact $|G|^{m}$-to-1. 

Given any $\phi \in \Hom(\pi_1(\Sigma), G)$, we construct a preimage $|g\rangle$ of $\phi$ as follows. The idea is that the group elements on the edges of the maximal spanning tree $T$ are arbitrary, but the group elements on the rest of the edges are completely determined in terms of these and $\phi$. For any edge $e$ not in $T$, let $\partial_0 e$ and $\partial_1 e$ be the two end vertices of $e$. By construction, there is a unique path $\gamma_i$ in $T$ connecting $v_0$ to $\partial_i e$, where $i = 0$, $1$. Let $\Bar{\gamma}_1$ be the path $\gamma_1$ with reversed direction, then $\gamma = \gamma_0 e \Bar{\gamma}_1$ is a closed path. An intuitive picture is that $\gamma$ reaches $\partial_0 e$ along $\gamma_0$ from $v_0$, travels through the edge $e$, and then goes back to $v_0$ along $\Bar{\gamma}_1$. There exists a unique group element $g_{e}$ such that 

\begin{equation}
    g_{\gamma_0} g_{e} g_{\Bar{\gamma}_1} = \phi(\gamma)
\end{equation}

It can be checked that $|g\rangle \in S$ and $\Phi(|g\rangle) = \phi$. Since we have $|G|^{m}$ choices of group elements to put on the spanning tree $T$ when defining $|g\rangle$, the map $\Phi$ is $|G|^{m}$-to-1.

On the other hand, for each given $|g\rangle$, if we are only allowed to apply gauge transformations on $|g\rangle$ at vertices other than $v_0$, there are in total $|G|^{m}$ such transformations. These transformations are all different from each other acting of a fixed $|g\rangle$.
If two basis elements $|g\rangle$ and $|g'\rangle$ are related by gauge transformations at vertices other than $v_0$, then $\Phi(|g\rangle) = \Phi(|g'\rangle)$. We conclude that the preimage of $\phi$ contains precisely those $|g\rangle$'s that are related by gauge transformations at vertices other than $v_0$. If we perform a gauge transformation $A_{v_0}(h)$ at $v_0$ to $|g\rangle$, then it is obvious that $\Phi(A_{v_0}(h)|g\rangle) = h \Phi(|g\rangle)h^{-1}$. Thus we have a one-to-one correspondence between gauge classes in $S$ and orbits in $\Hom(\pi_1(\Sigma), G)$.

\end{proof}

\section{Main results} \label{sec:main_result}
We now move on to the statement of our main theorem, which implies both TQO-1 and TQO-2.

\begin{theorem}
\label{thm:main}
Let $H$ be the Hamiltonian of  Kitaev's lattice model associated to any finite group $G$, closed surface $\Sigma$, and lattice $\Lambda$ on $\Sigma$. Let $A \subset B \subset \Lambda$ be two rectangular sublattices contained in contractible subregions such that $V(A) \subset \interior{V(B)}$, and denote 
\begin{align}\mathcal{H}_B = \{\ket{\psi} \in \mathcal{H}: A_v \ket{\psi} = B_p \ket{\psi} = \ket{\psi}, \forall v \in \interior{V(B)}, p \in F(B)\}.
\end{align}
Then all states in $\mathcal{H}_B$ have the same reduced state on $A$, i.e.\ there exists a density matrix $\rho_A$ on $A$ such that
\begin{equation}
    \Tr_{\bar{A}} \ket{\psi}\bra{\psi} = \rho_A,
\end{equation}
for all $\ket{\psi} \in \mathcal{H}_B$ such that $\langle \psi|\psi\rangle = 1$.
\end{theorem}

See Figure \ref{fig:subregion} for an example of the shape and arrangement of regions $A \subset B$.

 \begin{figure}[h] 
    \centering
    \includegraphics[scale=0.5]{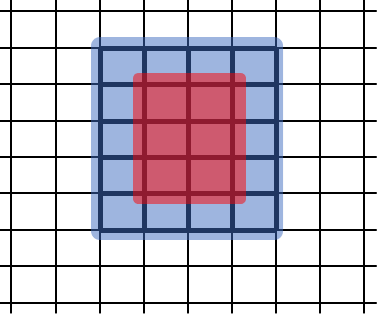}
      \caption{An example arrangement is shown of the regions $A \subset B$.  Region $A$ consists of the subset of edges in the red shaded rectangle, and likewise for $B$ in blue. Each region includes the edges in its rectangular boundary, as shown by the shading.} \label{fig:subregion}
\end{figure}

After warming up by proving that the toric code is a QECC in \S\ref{subsec:toric}, the main theorem is proved in Section \S\ref{subsec:indistinguishable}. In \S\ref{subsec:holonomies}, we point out a subtlety: we show that there exist choices of gauge groups for which the magnetic flux operators are insufficient data for specifying a gauge-invariant state, contrary to intuition from gauge theory based on e.g.\ special unitary groups.

\subsection{The toric code is a QECC: a warm-up}\label{subsec:toric}
In this section, we warm up by proving that the toric code is a QECC, which was shown in~\cite{kitaev2003fault}. The toric code is a special case of Kitaev's models, so this warm-up shows a weaker result than Theorem \ref{thm:main}, but we will improve upon both of these points in the next section.

We proceed by showing that the toric code obeys the Knill-Laflamme conditions, which state that a set of errors $\mathcal{E} = \left\{ E_i \right\}$ is correctable by an error correcting code represented by a projector $P$ onto the code subspace if and only if
\begin{equation}
  P E_i^\dagger E_j P = \alpha_{ij} P
\end{equation}
where $\alpha_{ij}$ form the entries of a Hermitian matrix (see Theorem \ref{kl_theorem}). The projection operator for the toric code is given by
\begin{equation}
  P = \prod_{v \in V} \frac{1 + X_v}{2} \prod_{p \in F} \frac{1+Z_p}{2}.
\end{equation}

Now, consider a general error on $k \in \N$ qubits. Since tensor products of Pauli operators span all possible operators, it is sufficient to consider the errors
\begin{equation}
  \mathcal{E}(k)\equiv \left\{ \bigotimes_{e \in S_k} \sigma_e : S_k \subseteq E, \abs{S_k} \leq k, \sigma_e \in \left\{ X, Y, Z \right\}\right\}.
\end{equation}
We claim that if $k > \lfloor \frac{L-1}{2} \rfloor$, where $L$ is the size of the lattice, then $\mathcal{E}(k)$ is not correctable. To see this, first note $k > \lfloor \frac{L-1}{2} \rfloor$ implies $k \ge \lceil \frac{L}{2} \rceil$. Thus, we can form the operator $E_i^\dagger E_j$ that is a tensor product of $X$ along a noncontractible loop by choosing appropriate $E_i, E_j$. Then, $E_i^\dagger E_j$ transforms two orthogonal states in the codespace into each other and therefore $P E_i^\dag E_j P \not \propto P$.

Now suppose $k \le \lfloor \frac{L-1}{2} \rfloor$. We first compute the commutation relations
\begin{align}
\begin{split}
  & (X \otimes I\otimes \dots \otimes I) \frac{1+ X_v}{2}  = \frac{1+ X_v}{2} (X \otimes I\otimes \dots \otimes I)\\
  & (X \otimes I\otimes \dots \otimes I) \frac{1+ Z_p}{2}  = \frac{1+ (-1)^{i(p)} Z_p}{2} (X \otimes I\otimes \dots \otimes I),
  \end{split}
\end{align}
where $I$ is the identity operator on $\mathbb{C}[\mathbb{Z}_2]$ and $i(p)$ is an indicator for whether the first edge is on the boundary of $p$. Similarly,
\begin{align}
\begin{split}
  & (Z \otimes I\otimes \dots \otimes I) \frac{1+ X_v}{2}  = \frac{1+ (-1)^{i(v)} X_v}{2} (Z \otimes I\otimes \dots \otimes I)\\
  & (Z \otimes I\otimes \dots \otimes I) \frac{1+ Z_p}{2}  = \frac{1+ Z_p}{2} (Z \otimes I\otimes \dots \otimes I).
  \end{split}
\end{align}
Now, we can represent, up to a phase, $E_i^\dag E_j$ as a product of Pauli's of the form $I \otimes \dots \otimes I \otimes \sigma \otimes I \otimes \dots \otimes I$, where $\sigma \in \left\{ X,Z \right\}$. We then commute $P$ across each of the factors. We first consider the edges on which a Pauli $Z$ is acted on. Then, unless every vertex is incident to an even number of them, there will exist a vertex $v$ for which $c(v)=1$, which would imply $P E_i^\dag E_j P = 0$. Otherwise, the edges form loops. Since there are at most $2k \le L-1$ edges acted on, the loops must be contractible. A similar argument holds for Pauli $X$ where there we work in the dual lattice. We conclude that $P E_i^\dag E_j P =  0$ or $E_i^\dag E_j$ is, up to a phase, a product of $X_v, Z_p$, which act trivially on the ground space. Hence the Knill-Laflamme condition is satisfied.

\subsection{States with locally zero energy density are locally indistinguishable}
\label{subsec:indistinguishable} 

We now give a proof of Theorem \ref{thm:main}.

Consider a rectangular subregion $A$, contained in a simply connected region of the surface.  (See Figure \ref{fig:subregion}.)
The rectangular assumption could be relaxed, at the cost of more complicated exposition.

Assume that some state $|\psi\rangle$ on the entire lattice is invariant under all $A_v, B_p$ operators whose support intersects with $A$.  That is, we assume
\begin{align}
\begin{split}
A_v |\psi\rangle & = |\psi\rangle \\
B_p |\psi\rangle & = |\psi \rangle
\end{split}
\end{align}
for all $A_v, B_p$ operators such that $v \in V(A)$ or $\partial p \cap E(A) \neq \emptyset$.  One can think of such a state as ``zero-energy density'' on the region $A$, where the energy density is given by the quantum double Hamiltonian.  

We will show that all such states $|\psi\rangle$ have the same reduced density matrix $\rho_A$ on the region $A$, and we will explicitly construct $\rho_A$.  

Recall that the ``group basis'' refers to the basis of states on the lattice of the form $\ket{g}\equiv \bigotimes_{e\in E}|g_e\rangle$.  For a such a basis state, it's helpful to introduce the notion of the holonomy around a loop.  Given a group basis state and an oriented closed loop on the lattice with fixed base point, the ``holonomy'' of the loop is the product of group elements on the edges, composed right to left in loop order.  While the choice of base point may generally affect the holonomy, in this proof we will only be concerned with the notion of loops with ``trivial holonomy,'' i.e.\ whose holonomy is the identity element. In that case, the base point does not matter, because the property of trivial holonomy does not depend on base point.  We discuss holonomies and their properties at greater length in Section 
\S\ref{subsec:holonomies}.  

Write $|\psi\rangle$ in the group basis. Because it is invariant under the $B_p$ operators intersecting $A$, the only product states in this expansion will be the ones with trivial holonomy on all closed loops in $A$.  Therefore we can write $|\psi\rangle$ as
\begin{align} \label{eq:decomposition1}
|\psi \rangle = \sum_{\substack{g_A \textrm{ with trivial} \\ \textrm{holonomies on } A}} |g_A\rangle_A |\phi_{g_A}\rangle_{\bar{A}},
\end{align}
where the sum is over all assignments $g_A = (g_e)_{e\in E(A)}$ of group elements to edges in $A$, such that all of the holonomies on $A$ are trivial.  The states $|\phi_{g_A}\rangle_{\bar{A}}$ are some set of states on $\bar{A}$ depending on $g_A$. Note that the states $|g_A\rangle_A $ are orthonormal, while the states $ |\phi_{g_A}\rangle_{\bar{A}}$ are not normalized and not necessarily orthogonal.

Next we will need the following result: For any two product states of group elements $|g_A\rangle$ and $|g_A'\rangle$ with trivial holonomies on $A$, and with the same group elements on the boundary $\partial A$, there exists a gauge transformation acting only on the vertices $V(A)^{\circ}$ in the interior of $A$ which transforms $|g_A\rangle_A$ to $|g_A'\rangle_A$.  That is, there is some gauge transformation $U_{\mathrm{int}}$ supported on the interior of $A$ such that 
\begin{align} \label{eq:interior_gauge_transf}
U_{\mathrm{int}} |g_A\rangle _A= |g_A'\rangle_A.
\end{align}

To build such a gauge transformation $U_{\mathrm{int}}$ that takes $|g_A\rangle_A$ to $|g_A'\rangle_A$ consider the internal vertices of $A$, ordered from left to right, then top to bottom.  Start at the top left internal vetex $v_0$, and choose the unique $g_0$ such that $A_{v_0}(g_0) |g_A\rangle_A$ matches $|g_A'\rangle_A$ on the entire top left plaquette.  Now move one vertex rightward, to internal vertex $v_1$, and choose the unique $g_1$ such that $A_{v_1}(g_1) A_{v_0}(g_0) |g_A\rangle_A$ matches $|g_A'\rangle_A$ on the top left two plaquettes.  Continue in this manner until we have found a gauge transformation on the top row of internal vertices such that both states match elements on the entire top row of plaquettes.  Repeat this procedure for the next row, and so on, until all internal vertices have been considered.  Then we have constructed the (unique) desired $U_{\mathrm{int}}$.

Consider any two terms $g_A, g_A'$ that appear in the decomposition \eqref{eq:decomposition1}.  By equation \eqref{eq:interior_gauge_transf}, there exists some gauge transformation $U_{\mathrm{int}}$ satisfying $U_{\mathrm{int}}|g_A\rangle_A  = |g_A'\rangle_A$. Since by assumption on the state $|\psi\rangle$, this gauge transformation leaves $|\psi\rangle$ invariant, we have 
\begin{align}
\begin{split}
|\phi_{g_A'}\rangle_{\bar{A}} & = {_A}\langle g_A' | \psi \rangle  \\
& =  {_A}\langle g_A' | U_{\mathrm{int}} | \psi \rangle \\
& =  {_A}\langle g_A |  \psi \rangle \\
& = |\phi_{g_A}\rangle_{\bar{A}} .
\end{split}
\end{align}

So $|\phi_{g_A'}\rangle_{\bar{A}} = |\phi_{g_A}\rangle_{\bar{A}}$ for any $g_A$, $g_A'$ with the same boundary elements.  We therefore take the subscript of $\phi$ to be an assignment of group elements to the edges of $\partial A$, e.g.\ $\phi_{g_{\partial A}}$ where $g_{\partial A} = (g_e)_{e\in E(\partial A)}$. Then, we can further refine the decomposition \eqref{eq:decomposition1} as 
\begin{align} \label{eq:decomposition2}
|\psi \rangle = \sum_{\substack{g_{\partial A} \textrm{ with trivial} \\ \textrm{holonomy on } \partial A}} |\xi_{g_{\partial A}} \rangle_A  |\phi_{g_{\partial A}}\rangle_{\bar{A}},
\end{align}
where the sum is over all assignments $g_{\partial A}$ of group elements to edges on the boundary $\partial A$, such that the holonomy on the entire boundary $\partial A$ is trivial.
The above decomposition uses the state 
\begin{align}
|\xi_{g_{\partial A}} \rangle_A \equiv \sum_{\substack{g_A \textrm{ with trivial }\\ \textrm{holonomy on }A \\ \textrm{s.t. } (g_A)|_{\partial A} = g_{\partial A}}} |g_A\rangle_A
\end{align}
where the sum is over all $g$ with trivial holonomy on $A$ whose elements on the boundary $\partial A$ match the assignment $g_{\partial A}$.

We will show that the decomposition of \eqref{eq:decomposition2} is actually a Schmidt decomposition, with uniform Schmidt coefficients, which have been absorbed into the non-normalized states $|\phi_{g_{\partial A}}\rangle_{\bar{A}}$.  We will show this by showing the states $|\phi_{g_{\partial A}}\rangle_{\bar{A}}$ are all orthogonal and have equal norm.

First, note that the states $|\phi_{g_{\partial A}}\rangle_{\bar{A}}$ and  $|\phi_{g'_{\partial A}}\rangle_{\bar{A}}$ are orthogonal for two distinct assignments $g_{\partial A}$ and $g'_{\partial A}$ of group elements to $\partial A$.  To see this, consider an edge $E_0 \in E(\partial A)$ on which  $g_{\partial A}$ and  $g'_{\partial A}$ differ.  We use the invariance of $|\psi\rangle$ under the operator $B_{P_0}$ for the plaquette $P_0$ that intersects $A$ precisely at this boundary edge $E_0$.  This invariance implies that both $|\phi_{g_{\partial A}}\rangle_{\bar{A}}$ and $|\phi_{g'_{\partial A}}\rangle_{\bar{A}}$ must be composed completely of product states in the group basis on $\bar{A}$ whose holonomy on $P_0$ (including the edge $E_0$) is trivial.  But because $g_{\partial A}$ and  $g'_{\partial A}$ differ at $E_0$, any two product states in the expansion of $|\phi_{g_{\partial A}}\rangle_{\bar{A}}$ and $|\phi_{g'_{\partial A}}\rangle_{\bar{A}}$ respectively must differ at some edge of $P_0$ in $\bar{A}$.  Thus $|\phi_{g_{\partial A}}\rangle_{\bar{A}}$ and $|\phi_{g'_{\partial A}}\rangle_{\bar{A}}$ must be orthogonal.

Next we will need the fact that for any two assignments $g_{\partial A}$ and $g'_{\partial A}$ of group elements to $\partial A$ with trivial holonomy on $\partial A$, there exists a gauge transformation $U_{\partial A}$ acting only on the vertices of $\partial A$ that brings $g_{\partial A}$ to  $g'_{\partial A}$.  To build such a gauge transformation $U_{\partial A}$, choose a contiguous ordering of the $L$ vertices of $\partial A$, starting with some vertex $v_0$.  Choose the unique $g_1$ such that $A_{v_1}(g_1)$ acting on the assignment $g_{\partial A}$ will match the assignment $g'_{\partial A}$ on the boundary edge from $v_0$ to $v_1$.  Next, choose the unique $g_2$ such that $A_{v_2}(g_2) A_{v_1}(g_1)$ acting on the assignment $g_{\partial A}$ will match the assignment $g'_{\partial A}$ on the boundary edges from $v_0$ to $v_2$.  Proceed in this way until finding a boundary gauge transformation $A_{v_{L-1}}(g_{L-1})\cdots A_{v_2}(g_2) A_{v_1}(g_1)$ acting on the assignment $g_{\partial A}$ matches the assignment $g'_{\partial A}$ on all the edges from $v_0$ to $v_{L-1}$.  Then these two assignments will also match on the final edge from $v_{L-1}$ to $v_0$, using the fact that both  $g_{\partial A}$ and  $g'_{\partial A}$ were assumed to have trivial holonomy along the boundary.  

To see that the states any two states $|\phi_{g_{\partial A}}\rangle_{\bar{A}}$ and $|\phi_{g'_{\partial A}}\rangle_{\bar{A}}$ have equal norm, consider the gauge tranformation $U_{\partial A}$ mentioned above, taking $g_{\partial A}$ to  $g'_{\partial A}$ .  We can factor $U_{\partial A}$ as a product of a unitary acting on $A$ and a unitary acting on $\bar{A}$, i.e.\ $U_{\partial A} = V_A V_{\bar{A}}$.  From the definition of $|\xi_{g_{\partial A}}\rangle_A$, we can see 
\begin{align}
V_A |\xi_{g_{\partial A}}\rangle_A = |\xi_{g'_{\partial A}}\rangle_A. 
\end{align}
Then from the invariance of $|\psi\rangle$ under $U_{\partial A}$, and using decomposition \eqref{eq:decomposition2}, we have 
\begin{align}
\begin{split}
|\phi_{g'_{\partial A}}\rangle_{\bar{A}} & = {_A}\langle \xi_{g'_{\partial A}} | \psi \rangle \\
& = {_A}\langle \xi_{g'_{\partial A}}| U_{\partial A} | \psi  \rangle \\
& =  V_{\bar A} ({_A}\langle \xi_{g'_{\partial A}}|  V_A  | \psi  \rangle) \\
& = \langle V_{\bar A} \xi_{g_{\partial A}}|  | \psi  \rangle \\
& = V_{\bar A} |\phi_{g_{\partial A}}\rangle_{\bar{A}}.
\end{split}
\end{align}
Because  $|\phi_{g_{\partial A}}\rangle_{\bar{A}}$ and $|\phi_{g'_{\partial A}}\rangle_{\bar{A}}$ are related by a unitary $V_{\bar A}$, it follows they must have the same norm.

We conclude that decomposition of \eqref{eq:decomposition2} is a Schmidt decomposition, with uniform Schmidt coefficients that have been absorbed into the non-normalized states $|\phi_{g_{\partial A}}\rangle_{\bar{A}}$.  Thus we can immediately calculate the reduced density matrix 
\begin{align}
\begin{split} \label{eq:rho_A_g}
\rho_A & = \Tr_{\bar A}{|\psi\rangle \langle \psi | } \\
& =\frac{1}{|G|^{|\partial A|-1}} \sum_{\substack{ g_{\partial A} \textrm{ with trivial}  \\ \textrm{ holonomy on } \partial A}} |\xi_{g_{\partial A}} \rangle_A  \langle\xi_{g_{\partial A}}|_A
\end{split}
\end{align}
where $|\partial A|$ is the number of boundary edges, and the normalization factor is calculated by noting there are $|G|^{|\partial A|-1}$ terms in the sum. (The latter is the number of assignments $g_{\partial A}$ of group elements to boundary edges such that the boundary has trivial holonomy: the group elements on all but one boundary edge may be chosen freely, and the element on the final edge is uniquely determined by requiring trivial holonomy.)

This reduced state on $A$ is manifestly invariant of the state $|\psi\rangle$, depending only on our original assumptions that $|\psi\rangle$ has zero energy-density on $A$.

\subsection{Wilson loops are not a complete set of observables}\label{subsec:holonomies}

It is standard in gauge theory to think of Wilson loop operators as the basic gauge-invariant observables. In typical models, such as gauge theory based on special unitary groups, or in Kitaev's lattice model based on the group $\mathbb{Z}_2$, these observables are sufficient to completely characterize a gauge-invariant state; see Sengupta \cite{sengupta1994gauge} for a discussion.  It is therefore tempting to think that this holds quite generally, e.g.\ for lattice gauge theory or Kitaev's lattice model based on any finite group $G$. Such a result would seem to suggest that Theorem \ref{thm:main} is ``morally obvious'': if gauge invariant states are determined by their Wilson loops, our main result would simply be an easy corollary of a local version of this statement. 

In fact, we will show that this naive intuition fails for certain choices of $G$. That is, for judicicously chosen $G$, we will exhibit a pair of \emph{orthogonal} gauge-invariant states with the same Wilson loops. This result emphasizes that it is a property only of the ground space---where the Wilson loops are not only locally the same, but also locally trivial---that states are determined by their non-Abelian fluxes.

The question of completeness for Wilson loop observables was addressed by Sengupta \cite{sengupta1994gauge} for continuous gauge groups.  To study quantum double models, we will be interested in the case of finite gauge groups and spatial lattices, which Sengupta comments on only briefly.

Let us state our claim more precisely. We will work in the gauge invariant subspace 
\begin{align}
    \mathcal{H}_{\mathrm{gauge}} = \{|\psi\rangle\in\mathcal{H}\mid A_v|\psi\rangle = |\psi\rangle, \ \forall v \in V \}.
\end{align}
The magnetic plaquete operators $B_{(v,p)}(h)$ do not in general preserve the gauge invariant subspace, so we will consider combinations which do:
\begin{align}
    B_{(v,p)}([h]) \equiv \frac{|[h]|}{|G|}\sum_{g\in G}B_{(v,p)}(ghg^{-1})
\end{align}
where $|[h]|$ is the order of the conjugacy class of $h$. These operators depend on $h$ only through its conjugacy class; heuristically, they compute the product of group elements around the plaquette and check whether or not that product is conjugate to $h$, annihilating the state if it is not, and stabilizing the state if it is. We are free to define more general magnetic operators $B_\gamma([h])$ for any closed loop $\gamma$, defined in the obvious way.  To avoid confusion, we will refer to the product of group elements around a path $\gamma$ with a base point (which is measured by the operator $B_\gamma(h)$ and is \emph{not} gauge-invariant) as a \emph{holonomy}; the conjugacy class of the product of group elements around a path $\gamma$ without a base point (which is measured by $B_\gamma([h])$ and \emph{is} gauge-invariant) will be referred to as a \emph{Wilson loop}.\footnote{Note that $B_\gamma(h)$ is a special case of the ribbon operator $F^{(\mathds{1},h)}(\gamma)$ associated to a closed ribbon $\gamma$, as defined in equation (24) of \cite{kitaev2003fault}. Therefore, Wilson loop operators $B_\gamma([h])$ are particular linear combinations of ribbon operators.}

Using e.g.\ equations \eqref{operator_properties}, it is straightforward to show that the $B_\gamma([h])$ commute with every $A_v$ and thus map $\mathcal{H}_{\mathrm{gauge}}\to \mathcal{H}_{\mathrm{gauge}}$. Moreover, they all commute with one another, and so we can work with a basis of the gauge-invariant subspace consisting of simultaneous eigenstates of the $B_\gamma([h])$. Our claim is then the following.

\begin{proposition}\label{prop:wilsonincomplete}
There exist Kitaev lattice models $(\Sigma,\Lambda,G)$ as well as an orthonormal pair of gauge-invariant states $|\psi\rangle$, $|\chi\rangle \in\mathcal{H}_{\mathrm{gauge}}$ which are eigenstates of every Wilson loop operator $B_\gamma([h])$ with identical eigenvalues, 
\begin{align}
    \langle\psi|B_\gamma([h])|\psi\rangle = \langle\chi|B_\gamma([h])|\chi\rangle.
\end{align} 
\end{proposition}
Thus, $|\psi\rangle$ and $|\chi\rangle$ are gauge-invariant states which cannot be distinguished by Wilson loop observables. The rest of this section is dedicated to the proof of this proposition.

The key ingredient which enters our construction is the existence of finite groups which admit \emph{outer class automorphisms}.\footnote{As we will see, this constitutes a sufficient condition on $G$ so that, for $\Lambda$ large enough, $(\Sigma,\Lambda, G)$ will satisfy Proposition \ref{prop:wilsonincomplete}.} An automorphism $\phi:G\to G$ is said to be outer if it is not of the form $\phi(g)= hgh^{-1}$ for any $h$ in $G$; it is a class automorphism if it preserves conjugacy classes, i.e.\ if $g$ is conjugate to $\phi(g)$ for every $g$ in $G$. We will not need any examples of such groups for our proof, so for the remainder of this section take their existence for granted. The interested reader is encouraged to consult e.g.\ \cite{wall1947finite} for examples of explicit constructions. 

Now, take $\Sigma$ any closed 2D surface, $G=\{h_1,\dots, h_N\}$ any finite group which admits an outer class automorphism $\phi$, and $\Lambda$ any lattice on $\Sigma$ which has at least an $N\times N$ square sublattice $A$, where $N:=|G|$. We start by constructing a state of $\mathcal{H}$ in the group basis which features every possible holonomy. In other words, we want a state $|g\rangle = \bigotimes_{e\in E(\Lambda)}|g_e\rangle$ such that, for each $h$ in $G$, there is some loop $\gamma$ such that $B_\gamma(h)|g\rangle = |g\rangle$. This is easy to achieve. Let $A_n\subset A$ for $n=1,\dots, |G|$ be the square of side-length $n$ whose lower-left corner sits at the lower left corner of $A$. Assign group elements to the edges of $A_1$ in such a way that that the product of elements around $A_1$ (starting at the lower-left corner of $A$) is equal to $h_1$. Proceed inductively by choosing group elements associated to the unassigned edges of $A_n$ in such a way that the holonomy around $A_n$ is equal to $h_n$. Finally, assign group elements to any remaining edges however you would like. The state $|g\rangle$ so constructed satisfies that $B_{A_n}(h_n)|g\rangle = |g\rangle$, and moreover $B_{A_n}([h_n])|g\rangle = |g\rangle$. 
Now, define $|\psi\rangle$ to be the ``gauge-symmetrization'' of $|g\rangle$,
\begin{align}\label{psistate}
    |\psi\rangle = \mathcal{N}_\psi\sum_{|g'\rangle\sim |g\rangle}|g'\rangle
\end{align}
where $\sim$ denotes gauge equivalence, and $\mathcal{N}_\psi$ is chosen to normalize $|\psi\rangle$. Gauge transformations at most change holonomies by conjugation, so $|\psi\rangle$ has the same Wilson loops as the state $|g\rangle$ from which it was constructed. We can do the same for the state $|\phi(g)\rangle= \bigotimes_{e\in E(\Lambda)}|\phi(g_e)\rangle$, and define 
\begin{align}\label{chistate}
    |\chi\rangle = \mathcal{N}_\chi \sum_{|g'\rangle\sim|\phi(g)\rangle}|g'\rangle
\end{align}
Since $\phi$ preserves conjugacy classes, $|\phi(g)\rangle$ and $|g\rangle$ have the same Wilson loops, and so it follows that $|\psi\rangle$ and $|\chi\rangle$ have the same Wilson loops as well. It remains only to show that these two states are orthogonal. We will do this by showing that $|g\rangle$ is gauge-inequivalent to $|\phi(g)\rangle$, from which it follows that every term in the sum in equation \eqref{psistate} is orthogonal to every term in the sum in equation \eqref{chistate}. For this we need the following lemma.
\begin{lemma}
Fix an arbitrary base-point $v_0$ in $V$. Two states $|g\rangle$, $|g'\rangle$ in the group basis are gauge equivalent if and only if their holonomies based at $v_0$ agree up to simultaneous conjugation, i.e.\ if and only if there is a single $h$ in $G$ such that $g'_\gamma = hg_\gamma h^{-1}$ for every loop $\gamma$ which starts and ends at $v_0$. 
\end{lemma}

\begin{proof}
In the forward direction, if $|g\rangle$ and $|g'\rangle$ are gauge-equivalent, they are by definition related by a product of gauge-transformations at different vertices, $|g'\rangle = \prod_{v\in V}A_v(h_v)|g\rangle$ for some $h_v$ in $G$. The gauge-transformations away from $v_0$ do not change the holonomies based at $v_0$, while the single gauge-transformation $A_{v_0}(h_{v_0})$ at $v_0$ changes all holonomies by conjugation by $h_{v_0}$, i.e.\ $g'_\gamma = h_{v_0}g_\gamma h_{v_0}^{-1}$. 

In the reverse direction, assume that the holonomies based at $v_0$ are simultaneously conjugate by an element $h$ in $G$. We will specify a sequence of gauge transformations which transforms $|g\rangle$ into $|g'\rangle$. First, we specify the gauge transformation needed at the base-point. Since acting with a gauge transformation at $v_0$ conjugates all holonomies based at $v_0$, we act with $A_{v_0}(h^{-1})$, so chosen because $A_{v_0}(h^{-1})|g\rangle$ will have the same based holonomies as $|g'\rangle$. 

Lay down a maximal spanning tree $T$ of $\Lambda$ which contains $v_0$. Recall that gauge transformations away from $v_0$ do not change holonomies based at $v_0$, and note that states in the group basis are fully determined by their holonomies based at $v_0$ as well as the group elements assigned to the edges of $T$. With this in mind, we will apply our remaining gauge transformations to the vertices of $T$ in order to make the states agree on the edges of $T$ (and therefore on all of $\Lambda$). 

We proceed inductively. Choose a path from $v_0$ to any leaf, and label the vertices which arise as $[v_0,v_1,\dots, v_r]$ and the edges between them as $[e_1,\dots, e_r]$. Compare the group element assigned to $e_1$ by $A_{v_0}(h^{-1})|g\rangle$ and $|g'\rangle$, and act with the unique gauge transformation $A_{v_1}(h_1)$ which makes $A_{v_1}(h_1)A_{v_0}(h^{-1})|g\rangle$ and $|g'\rangle$ agree at the edge $e_1$. Inductively walk through the path, and at the $n$th step, apply the unique gauge transformation $A_{v_n}(h_n)$ which makes $A_{v_n}(h_n)\cdots A_{v_1}(h_1)A_{v_0}(h^{-1})|g\rangle$ agree with $|g'\rangle$ at the edge $e_n$, noting that application of $A_{v_n}(h_n)$ does not interfere with any of the previously assigned edges $e_1,\dots, e_{n-1}$ since by assumption $T$ is a tree. The two states one has at the end of this procedure agree at all the edges $e_1,\dots, e_r$.

One can continue to apply this protocol to any remaining paths from a vertex in $T$ to one of its its leaves which have not yet been traversed. Calling the overall gauge transformation one obtains $\mathcal{G}$, the net result is that $\mathcal{G}|g\rangle$ and $|g'\rangle$ agree at every edge in the spanning tree, and have identical holonomies. It follows that they agree on all of $\Lambda$, and so they are in fact equal. 
\end{proof}

Now let $v_0$ be the lower-left hand corner of $A$. Recall that by construction, $g_{A_n} = h_n$ and meanwhile $\phi(g)_{A_n} = \phi(h_n)$. These holonomies cannot be the same up to simultaneous conjugation: since every element of $G$ is realized as a holonomy, this would imply that $\phi$ is not an outer automorphism, in contradiction with our assumption. Thus $|g\rangle$ is gauge-inequivalent from $|\phi(g)\rangle$, and so $|\psi\rangle$ is orthogonal to $|\chi\rangle$.

\subsection{Topological entanglement entropy} \label{subsec:TEE}
An important measure of topological order is topological entanglement entropy \cite{Kitaev:2005dm, levin2006detecting}.  The topological entanglement entropy refers to a particular term in the von Neumann entropy of the reduced state of the vacuum.  This term is well-defined whenever the entanglement entropy $S(\rho_A)$ for a simply connected region $A$ with boundary of length $L$ satisfies an ``area law,'' i.e.\ satisfies
\begin{align} \label{eq:S_topo_defn}
S(\rho_A) = \alpha L - \gamma + o(1),
\end{align}
where $\alpha, \gamma$ are some nonnegative constants.  In this case, the term $- \gamma$ is referred to as the topological entanglement entropy $S_{\textrm{topo}}$,
\begin{align} 
S_{\textrm{topo}} = - \gamma.
\end{align}
While this definition often works for simple models, the $L \to \infty$ limit on the lattice may not be well-defined for generic models.\footnote{Defining the $L \to \infty$ limit requires a family of regions of increasing size.  However, on the lattice, the boundaries of these regions will have ``sharp'' features like corners, and these sharp features make the limit sensitive to the choice of region.}  We therefore emphasize an alternative definition, which uses any three large, mutually adjacent regions $A$, $B$ and $C$ to define
\begin{align} \label{eq:S_topo_defn_regions}
    S_{\textrm{topo}} = S(\rho_{ABC}) - S(\rho_{AB})- S(\rho_{AC})- S(\rho_{BC}) + S(\rho_{A}) + S(\rho_{B}) + S(\rho_{C})
\end{align}
as in \cite{Kitaev:2005dm}.
This definition makes it clear that $S_{\textrm{topo}}$ is independent of any local contributions from ``sharp'' features like corners on the boundaries of the regions, since such contributions will cancel in equation \eqref{eq:S_topo_defn_regions}. Moreover, if we imagine a simple model in which entanglement entropy of every region obtains an additive contribution from each boundary edge as well as an overall additive constant $-\gamma$, then the precise definition in equation \eqref{eq:S_topo_defn_regions} matches the more heuristic definition in equation \eqref{eq:S_topo_defn}.

Assuming the system is gapped (so that the vacuum is expected to satisfy an area law in entanglement), one can more generally argue that $S_{\textrm{topo}}$ is independent of both deformations of regions $A,B,C$, and also perturbations of the Hamiltonian, so long as the system remains gapped \cite{Kitaev:2005dm}. It can therefore be used to characterize topological phases of matter.

After analyzing the ground state of the quantum double model in 
\S\ref{subsec:indistinguishable}, we are well-equipped to study the topological entanglement entropy.  Using the explicit form of the ground state reduced density matrix given in equation \eqref{eq:rho_A_g}, we find
\begin{align} 
S(\rho_A) = (L-1)\log{|G|} = L \log{|G|} - \log{|G|}.
\end{align}
By definition \eqref{eq:S_topo_defn} or \eqref{eq:S_topo_defn_regions}, we then find $S_{\textrm{topo}} = - \log{|G|}$.  

Recall that the sector of the quantum double model with $A_v=1$ may also be considered as a gauge theory: we consider the operators $A_v(g)$ as gauge transformations, and then we restrict our attention to the gauge-invariant Hilbert space, or ``physical'' Hilbert space, 
\begin{align} \label{eq:H_gauge}
\mathcal{H}_{\textrm{phys}} = \{\ket{\psi} \in \mathcal{H} \, : \, A_v \ket{\psi} = \ket{\psi} \, \text{for all } v \} .
\end{align}

In a gauge theory, unlike in an ordinary lattice theory, there are multiple, inequivalent ways to define the entanglement entropy of a region.  The ambiguity arises because when considering a region $A$ (consisting of some edges), the gauge-invariant Hilbert space $\mathcal{H}_{\textrm{phys}}$ does not factorize into a tensor product of factors associated with $A$ and the complement $\bar{A}$, unlike the full Hilbert space $\mathcal{H}$.  For a discussion of defining entanglement entropy in gauge theories, see \cite{acoleyen2015entanglement, soni2016aspects}.  Also see Appendix \ref{sec:EE_gauge} for an introduction.

Because there are multiple definitions for the full entanglement entropy in gauge theory, we can ask whether these definitions lead to the same answer for the \emph{topological} entanglement entropy. We will discuss three definitions in particular (introduced in the paragaph below equation \eqref{eq:EE_gauge}), including the ``ordinary'' definition which is calculated using the non-gauge-invariant/extended Hilbert space. We will find that two of these definitions --- the extended Hilbert space definition and the so-called ``algebraic'' definition --- differ by a term proportional to $L$, hence yielding the same answer for the term $S_{\textrm{topo}}$. This calculation corrects a comment made by \cite{soni2016aspects} and repeated elsewhere \cite{wong2018note, lin2018comments}, in the context of lattice gauge theories, stating that only one of these two definitions correctly reproduces the topological entanglement entropy.\footnote{We agree with the explicit calculations in the seminal work \cite{soni2016aspects}, but we rectify their incorrect suggestion that the algebraic definition of entanglement entropy does not accurately capture the topological entanglement entropy, e.g.\ when they state ``Dropping
some or all of these classical contributions would therefore not yield the correct result.'' Note that in \cite{soni2016aspects}, they refer to the algebraic definition of entanglement entropy as the ``electric centre'' definition.}  (The quantum double model coincides with the weakly coupled limit of a lattice gauge theory.)

We begin by considering reduced density matrices of general states in the gauge-invariant Hilbert space $\mathcal{H}_{\textrm{phys}}$. After examining the structure of the density matrix for a gauge-invariant state, we will be able to describe various possible definitions of entanglement entropy when the system is considered as a gauge theory.  

Consider a rectangular region $A$ (a subset of edges) and the reduced state $\rho_A$ on $A$ for some $\ket{\psi} \in \mathcal{H}_{\textrm{phys}}$. The region is rectangular for simplicity and includes the edges on its boundary, e.g.\ like region $A$ in Figure \ref{fig:subregion}.

Consider the boundary vertices $v \in \partial A$, and let $L=|\partial A|$ be the number of such vertices.  Then for each boundary vertex $v$, consider the gauge transformation $A_v(g)$ at $v$ but restricted to the edges from $v$ that lie in $A$. That is, $A_v(g)$ multiplies by $g$ (or $g^{-1}$, according to orientation)  for each edge that lies in $A$ of the boundary vertex $v$. We call this the ``boundary gauge transformation'' at $v$.  

For each boundary vertex $v$, there is a separate action of $G$ on the Hilbert space $\mathcal{H}_A$ of region $A$, where $\mathcal{H}_A$ is the associated tensor factor in the ``extended'' (non-gauge-invariant) Hilbert space.  The aforementioned actions commute.  We can therefore decompose $\mathcal{H}_A$ into representations of these actions as 
\begin{align} \label{eq:H_A_decomp}
\mathcal{H}_A = \bigoplus_{R_1,\dots,R_L \in \hat{G}} \mathcal{V}_{R_1}
\otimes \cdots \otimes \mathcal{V}_{R_L} \otimes \mathcal{W}_{\vec{R}}
\end{align}
where $\hat{G}$ denotes the set of inequivalent irreducible representations of $G$, and the direct sum is over all assignments $\vec{R}=(R_1,\dots,R_L)$ of irreps $R_i$ to boundary vertex $v_i$. The Hilbert space $\mathcal{V}_{R_i}$ denotes the Hilbert space of irrep $R_i$ of boundary gauge transformations at boundary vertex $v_i$, and the final Hilbert space factor $\mathcal{W}_{\vec{R}}$ is a multiplicity factor accounting for the multiplicity of representations of the boundary gauge transformations on $\mathcal{H}_A$.

Because $\ket{\psi} \in \mathcal{H}_{\textrm{phys}}$ is invariant under the full gauge transformation $A_v(g)$ at $v$, and because $A_v(g)$ factorizes into a part in $A$ and a part outside $A$, it follows that the boundary gauge transformations commute with the reduced state $\rho_A$. We can therefore decompose $\rho_A$ block-diagonally into the sectors of equation \eqref{eq:H_A_decomp} as 
\begin{align} \label{eq:rho_sector}
\mathcal{\rho}_A = \bigoplus_{R_1,\dots,R_L \in \hat{G}} p_{\vec{R}} \frac{\mathds{1}_{R_1}}{d_{R_1}}  \otimes \cdots \otimes \frac{\mathds{1}_{R_L}}{d_{R_L}}  \otimes \rho_{\vec{R}}
\end{align}
where $p_{\vec{R}}$ is a list of probabilities indexed by $(R_1,\dots,R_L)$, $\rho_{\vec{R}}$ is some state on $\mathcal{W}_{\vec{R}}$, and the factors $\frac{1}{d_{R_i}}$ are collected to ensure normalization of $\rho_{\vec{R}}$, with $d_R$ the dimension of irrep $R$.  The identity factors follow from Schur's lemma.

For a general gauge-invariant state $\ket{\psi}$, with $\rho_{\vec{R}}$ and $p_{\vec{R}}$ defined as above, the entanglement entropy (as computed in the full, non-gauge-invariant Hilbert space) is then
\begin{align} \label{eq:EE_gauge}
S(\rho_A) = -\sum_{\vec{R}} p_{\vec{R}}\log(p_{\vec{R}}) + \sum_{\vec{R}} p_{\vec{R}}\sum_{i=1}^L \log(d_{R_i}) + \sum_{\vec{R}} p_{\vec{R}}S(\rho_{\vec{R}}).
\end{align}

In Appendix \ref{sec:EE_gauge}, we introduce three definitions of entanglement entropy in gauge theory from the literature.  The first definition is simply the entanglement entropy as computed in the full non-gauge-invariant Hilbert space, as given by equation \eqref{eq:EE_gauge}.  This is sometimes called the ``extended Hilbert space definition.''  The second definition is sometimes called the ``algebraic'' definition, described by equation \eqref{eq:EE_algebraic} in Appendix \ref{sec:EE_gauge}. The algebraic definition corresponds to simply dropping the second term in equation \eqref{eq:EE_gauge}.  This term is sometimes called the ``$\log{\dim(R)}$ term,'' indicating its form.  Finally, a third definition of entanglement entropy drops both the first and second term of equation \eqref{eq:EE_gauge}.  This definition, leaving only the third term in equation \eqref{eq:EE_gauge}, gives the distillable entanglement  \cite{schuch2004nonlocal, soni2016aspects, acoleyen2015entanglement}.

In summary, various definitions entanglement entropy in gauge theory involve keeping all terms of equation \eqref{eq:EE_gauge}, keeping just the first and third term, or keeping just the third term.

The topological entanglement entropy as originally introduced by \cite{Kitaev:2005dm} uses the ``extended Hilbert space'' definition, i.e.\ uses the usual definition for entanglement entropy of a lattice theory, without special regard to the theory as a gauge theory. We regard the ``correct'' topological entanglement entropy as the value calculated this way, and we ask whether the other definitions of entanglement entropy in gauge theory also yield the ``correct'' topological entanglement entropy.  

We will find that in the quantum double model, the ``$\log{\dim(R)}$ term'' (i.e.\ the second term in equation \eqref{eq:EE_gauge}) is proportional to boundary length $L$, so it does not affect the calculation of the topological entanglement entropy.  In other words, for the case of the quantum double model, the ``algebraic'' definition of entanglement entropy does yield the correct topological entanglement entropy.  (On the other hand, considering only the distillable entanglement, i.e.\ the third term of equation \eqref{eq:EE_gauge}, does not yield the correct topological entanglement entropy.)

In fact, the ``$\log{\dim(R)}$ term'' may be seen as a sum over expectation values of the local observables $\log(d_{R_i})$.  These observables are local to vertices, for some choice of incoming edges. Their contribution then manifestly cancels in the definition of topological entanglement entropy given in \eqref{eq:S_topo_defn_regions}, in \emph{any} theory.\footnote{When the degrees of freedom live on edges, as in the quantum double model, the adjacent regions $A,B,C$ should be chosen so that they do not share edges.} However, since the contrary claim for the Kitaev double model has previously appeared in the literature, we shall take the time to calculate each term very explicitly.

To address the topological entanglement entropy of the ground state when using these various definitions of entanglement entropy in gauge theory, we simply need the reduced state $\rho_A$ of the ground state expressed in the form of equation \eqref{eq:rho_sector}.  Because the ground state is gauge-invariant, we must be able to put it in the form of equation \eqref{eq:rho_sector}, as explained above. Our goal is to explicitly obtain this form and then calculate the ``$\log{\dim(R)}$ term,'' i.e.\ the second term of equation \eqref{eq:EE_gauge}.

The following argument closely follows and elaborates Appendix B of \cite{acoleyen2015entanglement}. We will consider the trivial ground state of the quantum double model.  Because the topological entanglement entropy is computed locally, the following computation will also hold for other possible ground states, because these ground states look locally identical, as demonstrated by Theorem \ref{thm:main}.

Let $|\vec{\mathds{1}}\rangle$ be the product state in the group basis with the identity on every edge, define $P_A$ as the product of projectors $A_v$ for all vertices strictly in the interior of $A$ (so $P_A$ is local to $A$), and likewise define $P_{\bar{A}}$.  Note that we can write the trivial ground state as $|\psi_0\rangle \propto \prod_v A_v |\vec{\mathds{1}}\rangle$, which we can decompose as
\begin{align}
\begin{split}
|\psi_0\rangle & \propto P_{\bar{A}} \prod_{v \in \partial A} A_v P_A |\vec{\mathds{1}}\rangle \\
& \propto P_{\bar{A}} \sum_{g_1,\dots,g_L \in G} \prod_{i=1}^L A_{v_i}(g_i) P_A |\vec{\mathds{1}}\rangle 
\end{split}
\end{align}
where $v_1,\dots,v_n$ label the boundary vertices of $A$.  We can factor each $A_v(g)$ at boundary vertex $v$ into 
\begin{align} \label{eq:BGT_factorization}
    A_v(g)=A_v^A(g) A_v^{\bar{A}}(g),
\end{align} 
where $A_v^A(g)$ and $A_v^{\bar{A}}(g)$ are the parts of $A_v(g)$ that lie within $A$ and $\bar{A}$ respectively, i.e.\ $A_v^A(g)$ is the ``boundary gauge transformation'' discussed above, acting on $A$ from boundary vertex $v$.  Then
\begin{align} \label{eq:gs_A_ABar}
|\psi_0\rangle & \propto \sum_{g_1,\dots,g_L}  |\phi_{\bar{A}}(\vec{g})\rangle |\phi_A(\vec{g})\rangle
\end{align}
where $\vec{g}=(g_1,\dots, g_L)$ and we define the normalized states
\begin{align}
\begin{split}
    |\phi_A(\vec{g})\rangle & \propto \prod_{i=1}^L A^A_{v_i}(g_i) P_A |\vec{\mathds{1}}\rangle_A ,\\ 
    |\phi_{\bar{A}}(\vec{g})\rangle & \propto \prod_{i=1}^L A^{\bar{A}}_{v_i}(g_i) P_{\bar{A}} |\vec{\mathds{1}}\rangle_{\bar{A}}.
    \end{split}
\end{align}

Define an equivalence class $\vec{g} \sim \vec{h}$ on such lists when $\vec{g} = \vec{h}a$ for some $a \in G$.  Then $\langle \phi_A(\vec{h}) | \phi_A(\vec{g}) \rangle = 0$ when $\vec{g} \not \sim \vec{h}$.  To see this, first note $\vec{g} \sim \vec{h}$ if and only if the lists have the same ``differences,'' i.e.\ $g_i g_j^{-1} = h_i h_j^{-1}$ for all $i,j$.  Then note that $| \phi_A(\vec{g}) \rangle$ has group elements $g_i g_{i+1}^{-1}$ on its boundary edges, and if these elements do not match between $\vec{h}$ and $\vec{g}$ then $\langle \phi_A(\vec{h}) | \phi_A(\vec{g}) \rangle = 0$.  

Moreover, when $\vec{g} \sim \vec{h}$, then $
|\phi_A(\vec{h})\rangle =  | \phi_A(\vec{g}) \rangle$.  To see this, note that if $\vec{g} = \vec{h}a$, then
\begin{align}
\begin{split}
    | \phi_A(\vec{g}) \rangle &  = | \phi_A(\vec{h}a) \rangle \\
    & = \prod_{u\in A^\circ} A_u(a) | \phi_A(\vec{h}a) \rangle \\
    & =  \prod_{i=1}^L A^A_{v_i}(h_i) P_A  \prod_{i=1}^L  A_{v_i}(a) \prod_{u\in A^\circ} A_u(a)  |\vec{\mathds{1}}\rangle_A \\
      & =  \prod_{i=1}^L A^A_{v_i}(h_i) P_A |\vec{\mathds{1}}\rangle_A \\
      & =  | \phi_A(\vec{h}) \rangle 
\end{split}
\end{align}
recalling $v \in A^\circ$ denote vertices strictly in the interior of $A$. 

A nearly identical argument yields $ | \phi_{\bar{A}}(\vec{g}) \rangle =  | \phi_{\bar{A}}(\vec{h}) \rangle$ if $g \sim h$.  Moreover, if $g \not \sim h$, then $\langle \phi_{\bar{A}}(\vec{g})  | \phi_{\bar{A}}(\vec{h}) \rangle =0$.  To see this, note $| \phi_{\bar{A}}(\vec{g}) \rangle$ has a definite holonomy $g_i g_j^{-1}$ along an open path from boundary vertex $i$ to $j$, and if $g \not \sim h$, they must have some difference $g_i g_j^{-1} = h_i h_j^{-1}$, so they must have distinct holonomies, implying orthogonality.  

In summary, the state $| \phi_A(\vec{g}) \rangle , | \phi_A(\vec{h}) \rangle$ are equal for $\vec{g} \sim \vec{h}$, i.e.\ when $\vec{g}=\vec{h}a$ for some $a \in G$, and orthogonal otherwise; the same holds for states $| \phi_{\bar{A}}(\vec{g}) \rangle$ and $ | \phi_{\bar{A}}(\vec{h}) \rangle$.  Finally, note the states $| \phi_A(\vec{g}) \rangle , | \phi_A(\vec{h}) \rangle$ have equal norm, because they are related by a unitary (group multiplication), and likewise for the states  $| \phi_{\bar{A}}(\vec{g}) \rangle$ and $| \phi_{\bar{A}}(\vec{h}) \rangle$.

Then we can re-write equation \eqref{eq:gs_A_ABar} as
\begin{align} \label{eq:gs_A_ABar_Schmidt}
|\psi_0\rangle & \propto \sum_{\substack{g_1,...,g_{L-1} \\ g_L=\mathds{1}}}  |\phi_{\bar{A}}(\vec{g})\rangle |\phi_A(\vec{g})\rangle
\end{align}
noting that we can find a unique representative of each equivalence class of $\vec{g}$ by setting $g_L=\mathds{1}$.  By the above properties, equation \eqref{eq:gs_A_ABar_Schmidt} is also a Schmidt decomposition using terms of equal norm, and we immediately have 
\begin{align} \label{eq:rho_A_phi_decomp}
\Tr_{\bar{A}}(|\psi_0\rangle \langle \psi_0|) = \frac{1}{|G|^{L-1}} \sum_{\substack{g_1,...,g_{L-1} \\ g_L=\mathds{1}}}|\phi_A(\vec{g})\rangle \langle \phi_A(\vec{g})|.
\end{align}

Next we take advantage of a variant of the ``irrep'' basis discussed in Appendix \ref{sec:irrep_basis} and associated orthogonality relations.  See the appendix for an introduction to these methods.  Again following \cite{acoleyen2015entanglement}, we define states 
\begin{align}
|\vec{R},\vec{i},\alpha\rangle = \frac{1}{|G|^L} \sum_{\vec{g}} \sum_{\vec{j}} c^\alpha_{\vec{R},\vec{j}}D^{\vec{R}}_{\vec{i},\vec{j}}(\vec{g})^*|\phi_A(\vec{g})\rangle, 
\end{align}
which bear some explaining. The list $\vec{R}=(R_1,\dots,R_L)$ is any list of irreps of $G$, with irrep $R_i$ associated to boundary vertex $v_i \in \partial A$.  The list $\vec{g}=(g_1,\dots,g_L)$ is any list of group elements at each boundary vertex.  We fix a basis for each distinct irrep, and the component $i_k$ of the list $\vec{i}$ is an index for the basis of the irrep $R_k$.  The quantities $D^{R}_{i,j}(g) = \langle i|D^R(g)|j\rangle$ are the representation matrices, namely the $(i,j)$-entry of the element $g$ under the representation $R$, and we define a shorthand notation for the product
\begin{align}
    D^{\vec{R}}_{\vec{i},\vec{j}}(\vec{g}) := D^{R_1}_{i_1,j_1}(g_1)\cdots D^{R_L}_{i_L,j_L}(g_L).
\end{align} The index $\alpha$ labels copies of the singlet (trivial) representation appearing in the tensor product of irreps $R_1,\dots,R_L$.  These singlet representations are encoded by the tensor $c^{\alpha}_{\vec{R},\vec{j}}$. In particular, for fixed $\alpha$ and $\vec{R}$, $c^{\alpha}_{\vec{R},\vec{j}}$ gives the associated $G$-invariant tensor (which spans the associated singlet representation), satisfying %
\begin{align}
    \sum_{\vec{k}} D^{R_1 }_{j_1 k_1}(g)\cdots D^{R_L }_{j_L k_L}(g) c^\alpha_{\vec{R},\vec{k}} = c^\alpha_{\vec{R},\vec{j}}
\end{align} for all $g \in G$. 

The significance of the states $|\vec{R},\vec{i},\alpha\rangle$ are their simple transformation properties under boundary gauge transformations.  For a boundary gauge transformation $A_{v_1}^A(g_1)\cdots A_{v_L}^A(g_L)$ acting at boundary vertices $v_1,\dots,v_L$ with gauge transformations $g_1,\dots,g_L$ (recalling equation \eqref{eq:BGT_factorization}), direct computation shows \begin{align}
    A_{v_1}^A(g_1)\cdots A_{v_L}^A(g_L) |\vec{R},\vec{i},\alpha\rangle = \sum_{\vec{j}} D^{\vec{R}}_{\vec{j},\vec{i}}(\vec{g}) |\vec{R},\vec{j},\alpha\rangle.
\end{align}  In other words, for fixed $\vec{R}$ and $\alpha$, the states $|\vec{R},\vec{i},\alpha\rangle$ furnish the representation $R_1,\dots,R_L$ of the boundary gauge transformations.  Accordingly, for fixed $\vec{R}$ and $\alpha$, these states furnish the factors $\mathcal{V}_{R_1}
\otimes \cdots \otimes \mathcal{V}_{R_L}$ in equation  \eqref{eq:H_A_decomp}.

If we choose the invariant tensors $c^{\alpha}_{\vec{R},\vec{j}}$ to be normalized as 
\begin{align}
    \sum_{\vec{j}} c^{\alpha *}_{\vec{R},\vec{j}} c^{\beta}_{\vec{R},\vec{j}} = d_{\vec{R}}\delta_{\alpha,\beta},
\end{align} with $d_{\vec{R}}= \prod_i d_{R_i}$, then it follows the states $|\vec{R},\vec{i},\alpha\rangle$ are orthonormal.  The orthonormality may be computed directly,
\begin{align}
\begin{split}
\langle \vec{R}',\vec{i}',\alpha'| \vec{R},\vec{i},\alpha\rangle & = \sum_{\vec{g},\vec{g}', \vec{j},\vec{j}'} c^{\alpha' *}_{\vec{R}',\vec{j}'} c^\alpha_{\vec{R},\vec{j}} D^{\vec{R}'}_{\vec{i}',\vec{j}'}(\vec{g}') D^{\vec{R}}_{\vec{i},\vec{j}}(\vec{g})^*  \langle \phi_A(\vec{g}') | \phi_A(\vec{g}) \rangle \\
& = \sum_{h \in G, \vec{g}, \vec{j},\vec{j}'} c^{\alpha' *}_{\vec{R}',\vec{j}'} c^\alpha_{\vec{R},\vec{j}} D^{\vec{R}'}_{\vec{i}',\vec{j}'}(\vec{g})  D^{\vec{R}}_{\vec{i},\vec{j}}(\vec{g}h)^* \langle \phi_A(\vec{g}h) | \phi_A(\vec{g}) \rangle  \\
&= \sum_{\vec{g}, \vec{j},\vec{j}'} c^{\alpha' *}_{\vec{R}',\vec{j}'} c^\alpha_{\vec{R},\vec{j}} D^{\vec{R}'}_{\vec{i}',\vec{j}'}(\vec{g}) D^{\vec{R}}_{\vec{i},\vec{j}}(\vec{g})^*  \\
& = \sum_{\vec{j}} c^{\alpha' *}_{\vec{R}',\vec{j}'} c^\alpha_{\vec{R},\vec{j}} \delta_{\vec{R},\vec{R}'} \delta_{\vec{i},\vec{i}'}  \\
& = \delta_{\vec{R},\vec{R}'} \delta_{\vec{i},\vec{i}'} \delta_{\alpha,\alpha'}.  
\end{split}
\end{align}
From first to second line, we used the fact that $|\phi_A(\vec{g})\rangle$ and $|\phi_A(\vec{g}')\rangle$ are equal if $\vec{g}' = \vec{g}h$ for some $h \in G$ and orthogonal otherwise.  From second to third line, we use the invariance of the $c^{\alpha}_{\vec{R},\vec{j}}$ tensor under $h$, discussed above.  From third to fourth, we use the orthogonality relation of  \eqref{eq:schur_orthogonality}, and finally from fourth to fifth we use the orthonormality of the $c^{\alpha}_{\vec{R},\vec{j}}$ tensors.

It will be helpful to count the number of states $|\vec{R},\vec{i},\alpha\rangle$.  We have
\begin{align}
| \{ |\vec{R},\vec{i},\alpha\rangle \} | = \sum_{\vec{R}} d_{\vec{R}} N_{\vec{R}}
\end{align}
where $N_{\vec{R}}$ is the dimension of the invariant subspace of the tensor product representation $R_1,\dots,R_L$, i.e.\ the number of copies of the singlet representation, corresponding to the distinct values of $\alpha$, and the factor $d_{\vec{R}}$ arises from the possible values of $\vec{i}$.  Note that the number of singlets $N_{\vec{R}}$ occurring in the fusion of $R_1,\dots,R_L$ is precisely the number of copies of $R_L^*$ occurring in the fusion of $R_1,\dots,R_{L-1}$.  Moreover, for any representation $Q$ (not necessarily an irrep), we have 
\begin{align} \label{eq:rep_fact}
    \sum_r d_r N_{r \to Q} = d_Q
\end{align} where the sum is over all irreps $r$ and $N_{r \to Q}$ denotes the number of copies of irrep $r$ in $Q$.  Using these facts, we have
\begin{align} 
\begin{split}
    \sum_{\vec{R}} d_{\vec{R}} N_{\vec{R}} & = \sum_{R_1,\dots,R_{L-1}} d_{R_1}\cdots d_{R_{L-1}}  \sum_{R_L} d_{R_L} N_{\vec{R}} \\
    & = \sum_{R_1,\dots,R_{L-1}} (d_{R_1}\cdots d_{R_{L-1}})^2 \\
    & = |G|^{L-1}
\end{split}
\end{align}
where we used the above facts to obtain the second line from the first; the final line follows from the identity $\sum_R d_R^2 = |G|$.  

Having counted $|G|^{L-1}$ orthonormal states $ |\vec{R},\vec{i},\alpha\rangle$, which are in the span of the states $|\phi_A(\vec{g})\rangle$, and noting $\dim\textrm{span} \{|\phi_A(\vec{g})\rangle\} = |G|^{L-1}$, we conclude that the states $ |\vec{R},\vec{i},\alpha\rangle$ form an orthonormal basis for $\textrm{span} \{|\phi_A(\vec{g})\rangle\}$.  From equation \eqref{eq:rho_A_phi_decomp}, we know that the vacuum reduced state $\rho_A$ is the maximally mixed state supported on the latter subspace, so we can change basis and conclude
\begin{align}
    \Tr_{\bar{A}}(|\psi_0\rangle \langle \psi_0|) = \frac{1}{|G|^{L-1}}\sum_{\vec{R}}\sum_{\vec{i},\alpha} |\vec{R},\vec{i},\alpha \rangle \langle \vec{R},\vec{i},\alpha|.
\end{align}
It follows that
\begin{align}
p_{\vec{R}} =\frac{1}{|G|^{L-1}} N_{\vec{R}} d_{\vec{R}}.
\end{align}
The ``$\log{\dim{R}}$'' term is
\begin{align}
\begin{split}
\sum_{\vec{R}} p_{\vec{R}} \sum_{i=1}^L \log{d_{R_i}}&  = 
\frac{1}{|G|^{L-1}} \sum_{\vec{R}}  N_{\vec{R}} d_{\vec{R}} \log{d_{\vec{R}}}  \\
& = \frac{1}{|G|^{L-1}} \sum_{i=1}^L \sum_{R_i} \log{d_{R_i}} \sum_{\{R_j\}_{j, j\neq i}}  N_{\vec{R}} d_{\vec{R}}.
\end{split}
\end{align}
For each $i$, also choose some boundary vertex index $a_i \in \{1,\dots,L\},\, a_i \neq i$. Below, we will also write $a_i$ as simply $a$. Then the rightmost sum, which depends on some fixed $i$ and $R_i$, may be re-written as 
\begin{align}
\begin{split}
\sum_{\{R_j\}_{j, j\neq i}}  N_{\vec{R}} d_{\vec{R}} & = \sum_{\{R_j\}_{j, j \neq i,a}} \sum_{R_a} d_{\vec{R}} N_{\vec{R}} \\
& = \sum_{\{R_j\}_{j, j \neq i,a}} \displaystyle\prod_{l \neq a} d_{R_l}  \sum_{R_a} d_{R_a} N_{\vec{R}}.
\end{split}
\end{align}
Then the new rightmost sum may be re-written using the identity in equation \eqref{eq:rep_fact} as
\begin{align}
\sum_{R_a} d_{R_a} N_{\vec{R}} = \displaystyle\prod_{m \neq a} d_{R_m}
\end{align}
so that we obtain
\begin{align}
\begin{split}
\sum_{\{R_i\}_{i, i\neq j}}  N_{\vec{R}} d_{\vec{R}}  &=   \sum_{\{R_j\}_{j, j \neq i,a}} \displaystyle\prod_{l \neq a} d_{R_l}^2 \\
& = d_{R_i}^2 |G|^{L-2}
\end{split}
\end{align}
and then, plugging back in,
\begin{align}
\begin{split}
\sum_{\vec{R}} p_{\vec{R}} \sum_{i=1}^L \log{d_{R_i}}& = \frac{1}{|G|} \sum_{i=1}^L \sum_{R_i} \log{d_{R_i}} d_{R_i}^2 \\
& = \frac{L}{|G|} \sum_{R \in \hat{G}} d_R^2 \log{d_R}
\end{split}
\end{align}
which is proportional to $L$, as claimed.  

For completeness, let us conclude by listing the entanglement entropy for the rectangular region for each of the three definitions. The distillable entropy is
\begin{align}
    \sum_{\vec{R}} p_{\vec{R}} S(\rho_{\vec{R}}) = \frac{1}{|G|^{L-1}} \sum_{\vec{R}} N_{\vec{R}} d_{\vec{R}} \log N_{\vec{R}},
\end{align}
as found in \cite{acoleyen2015entanglement}, where $\rho_{\vec{R}}$ is defined implicitly by equation \eqref{eq:rho_sector}. In general, this does not give the correct TEE; it is always zero in Abelian theories, for example.
The algebraic entropy is
\begin{align}
\begin{split}
    L\left[ \log |G| - \frac{1}{|G|}\sum_{R \in \hat{G}} d_R^2 \log d_R \right] - \log |G|.
    \end{split}
\end{align}
Finally, the extended Hilbert space entropy is
\begin{align}
    (L-1) \log |G|.
\end{align}

\section{Conclusion}

In this note, we have shown that Kitaev's finite group models obey a theorem which, at the level of slogans, says that ``states with locally zero energy density are locally indistinguishable.'' The theorem implies in particular that Kitaev's models have topological quantum order (TQO-1 and TQO-2) and moreover furnish a quantum error correcting code (QECC), a fact which, although well-appreciated, appears not to have been proved rigorously in the literature. In contrast, we have demonstrated that an analogous result cannot hold for excited states. Namely, contrary to intuitions one might have from typical gauge theory models, Wilson loop operators do not form a complete set of commuting observables.  Finally, we have also used our detailed analysis of the ground states to analyze the topological entanglement entropy, in particular demonstrating that the algebraic definition of entanglement entropy yields the correct topological entanglement entropy, correcting claims made previously in the literature. 

As was mentioned in the introduction, Kitaev's models can be generalized from finite groups to Hopf $\mathbb{C}^\ast$-algebras; the latter reduce to the former when the Hopf-algebra is taken to be the group algebra $\mathbb{C}[G]$ associated to $G$. It is interesting to ask to what extent the techniques we have used can be adapted to the Hopf-algebra case. Given the equivalence between these generalized Kitaev models and the Levin-Wen string net models, a successful generalization would therefore constitute a proof that the Levin-Wen models have TQO as well. We leave this question for future study.\\

\noindent \textbf{Acknowledgments.} SC would like to thank Zhenghan Wang for bringing the problem studied in this paper to his attention. We would like to thank Matthew Hastings, Patrick Hayden, Alexei Kitaev, and Xiao-Liang Qi for their helpful discussions. BR gratefully acknowledges NSF grant PHY 1720397. DD would like to thank God for all of His provisions and is supported by a National Defense Science and Engineering Graduate Fellowship. XH is supported by a Stanford Graduate Fellowship. SC acknowledges support from Simons Foundation, Virginia Tech, and Purdue University.

\appendix
\section{The ``irrep basis'' and the non-Abelian Fourier transform} \label{sec:irrep_basis}

We restrict our attention to finite groups, for the purpose of the quantum double model, but analogous results hold for compact groups.  Consider the Hilbert space $L^2(G) = \textrm{span}\{\ket{g} \, : \, g \in G\}$.  The group $G$ acts on $L^2(G)$ in at least two natural and distinct ways: left multiplication by $g$, i.e.\ as $L_g\ket{h} \equiv \ket{gh}$, and right multiplication by $g^{-1}$, i.e.\ as $R_g \ket{h} \equiv \ket{hg^{-1}}$.  These actions commute, and therefore we can simultaneously decompose the Hilbert space into irreducible representations of both actions.  In particular, it turns out the Hilbert space decomposes as 
\begin{align}
L^2(G)  = \bigoplus_{\mu \in \hat{G}} \mathcal{H}_{\mu} \otimes \mathcal{H}_{\mu^*}
\end{align}
where the sum is over all inequivalent irreducible representations (``irreps'') $\mu$, with the set of irreps denoted by $\hat{G}$.  The Hilbert space of an irrep $\mu$ is denoted by $\mathcal{H}_{\mu}$, and $\mu^*$ denotes the dual representation.  In the above decomposition, action of $G$ by left multiplication acts non-trivially only on the left factors $\mathcal{H}_\mu$, whereas the action of $G$ by right multiplication (of $g^{-1}$) acts non-trivially only on the right factors $\mathcal{H}_{\mu^*}$.  The above decomposition is not necessarily obvious, but it is a basic fact in the representation theory of finite groups.

Compatible with the above decomposition, we can define the orthonormal basis $\ket{\mu; a,b}$ of $L^2(G)$, where $\mu$ runs over all irreps, where $a=1,\dots,d_\mu$ is a label of some basis of irrep $\mu$ with dimension $d_\mu$, and where $b=1,\dots,d_\mu$ is a label of the corresponding dual basis of $\mu^*$, noting $d_\mu=d_{\mu^*}$. This is called the ``irrep basis.'' 

A foundational fact in the representation theory of finite groups yields the following transformation between orthonormal bases $\ket{g}$ and $\ket{\mu;a,b}$:
\begin{align} \label{eq:irrep_basis_defn}
\begin{split}
\ket{\mu;a,b} & = \sqrt{\frac{d_\mu}{|G|}} \sum_{g \in G} D^\mu_{ab}(g) \ket{g} \\
\ket{g} & = \sum_{\mu \in \hat{G}} \sum_{a,b=1}^{d_\mu} \sqrt{\frac{d_\mu}{|G|}}  D^\mu_{ab}(g)^* \ket{\mu;a,b}
\end{split}
\end{align}
where again $\hat{G}$ is the set distinct irreps of $G$, where $d_\mu$ is the dimension of irrep $\mu$, and where $D^\mu_{ab}(g)$ denotes the $a,b$ matrix element of the operator $D^\mu(g)$ that represents $g$ under irrep $\mu$.  (For instance, if $G=SU(2)$, and $\mu$ is chosen as the spin-$\frac{1}{2}$ representation, then $D^\mu(g)$ would be a $2 \times 2$ matrix with indices $a,b$, and it could be calculated using the Pauli matrices.) The above transformation may be considered a consequence of the Peter-Weyl theorem for finite groups.   

To be concrete, the irrep basis is defined by choosing a basis of each irrep $\mu$ of $G$ and then using equation \eqref{eq:irrep_basis_defn} to express $\ket{\mu;a,b}$ in terms of states $\ket{g}$.  Given the orthonormality of the $\vec{g}$ basis, the orthonormality of the $\ket{\mu;a,b}$ basis lies in the orthogonality relation (the Schur ``grand'' orthogonality relation)
\begin{align} \label{eq:schur_orthogonality}
\frac{d_\mu}{|G|} \sum_{g \in G} D^{\mu' *}_{a'b'}(g) D^\mu_{ab}(g) =  \delta_{\mu\mu'} \delta_{aa'}\delta_{bb'} 
\end{align}

The states of gauge theories expressed in the irrep basis are also called ``spin networks'' \cite{baez1996spin}, and they are related to string nets \cite{buerschaper2009mapping}.  For a thorough description of the quantum double model in the irrep basis, see \cite{buerschaper2009mapping}, which relates the quantum double model to the string net models of Levin-Wen \cite{levin2005string}.

One may calculate that the gauge-invariant Hilbert space of the quantum double model decomposes as (see equation \eqref{eq:H_gauge}, Section \ref{subsec:TEE}) 
\begin{align} \label{eq:H_gauge_decomp}
\mathcal{H}_{\textrm{phys}} = \bigoplus_{\mu(e)} \bigotimes_v \textrm{Singlet}\left( \bigotimes_{e_v \in \textrm{Edges}(v)} \mathcal{H}_{\mu(e_v)} \right),
\end{align}
where the direct sum is over all assignments of irreps $\mu(e)$ to all oriented edges $e$.  The tensor product $\bigotimes_v$ is over all vertices $v$.  The notation $\textrm{Singlet}(R)$ for some (possibly non-irreducible) representation $R$ denotes the singlet subsector of $R$, i.e.\ the subspace of trivial representations. (The latter is also denoted $\textrm{Hom}(1,R)$.) The tensor product $\bigotimes_{e_v \in \textrm{Edges}(v)}$ is over all edges $e_v$ connected to vertex $v$, and finally $\mathcal{H}_{\mu(e_v)}$ is the Hilbert space associated to the irrep $\mu(e_v)$ assigned to edge $e_v$, with dual representations taken for ingoing edges.  

The calculation that the gauge-invariant Hilbert space decomposes this way amounts to unpacking of notation.  

\section{Appendix: Entanglement entropy in gauge theory and the algebraic formalism} \label{sec:EE_gauge}

In an ordinary lattice theory, the Hilbert space may be expressed as a tensor product of lattice degrees of freedom (on vertices or edges).  The algebra of operators on the Hilbert space then factorizes also, allowing for a natural definition of a partial trace, a reduced density matrix, and a von Neumann entropy, used to define entanglement entropy.

Meanwhile, when viewing a lattice theory like the quantum double model as a gauge theory, we want to restrict our attention to the gauge-invariant or ``physical'' Hilbert space $\mathcal{H}_{\textrm{phys}} \subset \mathcal{H}$, as in equation \eqref{eq:H_gauge}, forgetting the embedding into the larger ``unphysical'' or ``extended'' Hilbert space $\mathcal{H}$ of the lattice.  Let $L(\mathcal{H})$ denote the algebra of linear operators on $\mathcal{H}$, and likewise for $L(\mathcal{H}_{\textrm{phys}})$. Then $L(\mathcal{H}_{\textrm{phys}})$ are the ``physical'' or gauge-invariant observables, and we can also think of them as the subset of $L(\mathcal{H})$ that commutes with gauge transformations, subsequently restricted to $\mathcal{H}_{\textrm{phys}}$.   

For any region $A$ of the lattice, we can define a sub-algebra $\mathcal{A} \subset L(\mathcal{H}_{\textrm{phys}})$ of the gauge-invariant observables associated to region $A$.  In particular, we define the operators in $\mathcal{A}$ as (the restriction to $\mathcal{H}_{\textrm{phys}}$ of) all operators in $L(\mathcal{H})$ that are local to $A$ and commute with all gauge transformations.  

In general, given a region $A$, although we can define the local gauge-invariant observables $\mathcal{A}$, the full algebra of gauge-invariant operators $L(\mathcal{H}_{\textrm{phys}})$ will not factorize as $L(\mathcal{H}_{\textrm{phys}}) \stackrel{?}{=} \mathcal{A} \otimes \mathcal{B}$ for any complementary factor $\mathcal{B}$.\footnote{For a paradigmatic example demonstrating that no factorization exists, consider a Wilson loop that straddles the boundary of $A$, crossing from $A$ to $\bar{A}$ and back.  Such an operator cannot be written as a sum of products of gauge-invariant operators on $A$ and $\bar{A}$, so evidently the tensor product of the gauge-invariant algebras of observables on $A$ and $\bar{A}$ is a proper subspace of the entire gauge-invariant algebra.}
Therefore, the usual definitions of partial trace, reduced density matrix, and entanglement entropy must be modified.  

One option is to work entirely within the ``extended'' or ``unphysical'' Hilbert space $\mathcal{H}$ and calculate entanglement entropy according to the usual definition.  However, this definition may be undesirable in the context of gauge theory, where one often wants to view the theory as intrinsic to $\mathcal{H}_{\textrm{phys}}$, without reference to some embedding.  

A second option for defining entanglement entropy in gauge theories is to use the ``algebraic'' definition of entanglement entropy, applied to the gauge-invariant Hilbert space and observables.  In this algebraic framework \cite{ohya2004quantum, casini2014remarks}, one defines a reduced density matrix and entanglement entropy for any state, taken with respect to a sub-algebra such as $\mathcal{A}$.  The following brief review follows Section 3 of \cite{mazenc2019target}.  We will assume a finite-dimensional Hilbert space throughout.  To define the entanglement entropy, first note that for general Hilbert space $\mathcal{V}$ and algebra $\mathcal{A} \subset L(\mathcal{V})$, there exists a decomposition of the Hilbert space as a direct sum of tensor products,
\begin{equation} \label{eqn:V_decomposition}
\mathcal{V} = \bigoplus_i \mathcal{V}_{i} \otimes \widetilde{\mathcal{V}}_{i} 
\end{equation}
such that the operators $\mathcal{O}_A \in \mathcal{A}$ are precisely those which take the form
\begin{equation}
\mathcal{O}_A = \sum_i \mathcal{O}_{A,i} \otimes \mathds{1}_{\widetilde{\mathcal{V}}_{i}}
\end{equation}
for some $\mathcal{O}_{A,i} \in L(\mathcal{V}_i)$.  (This follows from an elementary version of the Artin-Wedderburn theorem.) Schematically, we can also write
\begin{align}
\mathcal{A} = \bigoplus_i  L(\mathcal{V}_{i}) \otimes \mathds{1}_{\widetilde{\mathcal{V}}_{i}}.
\end{align}
Define $\Pi_i$ as the projectors onto the sectors $\mathcal{V}_{i} \otimes \widetilde{\mathcal{V}}_i$ in equation \eqref{eqn:V_decomposition}.  For any state $\ket{\psi} \in \mathcal{V}$, note $\Pi_i \ket{\psi}\bra{\psi} \Pi_i \in L(\mathcal{V}_{i}) \otimes L(\widetilde{\mathcal{V}}_i)$, and define
\begin{align}
\begin{split}
\rho_{A,i}& = \Tr_{\widetilde{\mathcal{V}}_i}(\Pi_i \ket{\psi}\bra{\psi}) \\
p_i &= \Tr(\Pi_i \ket{\psi}\bra{\psi})
\end{split}
\end{align}
Then one defines the entanglement entropy $S(\ket{\psi},\mathcal{A})$ of $\ket{\psi}$ with respect to sub-algebra $\mathcal{A}$ as 
\begin{align} \label{eq:EE_algebraic}
S(\ket{\psi},\mathcal{A}) = - \sum_i p_i \log(p_i) + \sum_i p_i S(\rho_{A,i}).
\end{align}
An equivalent definition is 
\begin{align}
    S(\ket{\psi},\mathcal{A}) = S(\mathcal{P}_{\mathcal{A}}(\ket{\psi}\bra{\psi})),
\end{align}
where the superoperator $\mathcal{P}_{\mathcal{A}}$ projects the density matrix $\ket{\psi}\bra{\psi}$ into the algebra $\mathcal{A}$.  (However, when using this shorthand definition, the image of the projection must be considered to lie in the abstract algebra $\mathcal{A}$, rather than inside its embedding $\mathcal{A} \subset L(\mathcal{V})$; the latter has extra identity factors, and the entropy of these factors is not included in the algebraic definition of the entropy.)

The first and second term of \eqref{eq:EE_algebraic} are sometimes referred to as the ``classical'' and ``quantum'' pieces, respectively.  In the context of gauge theory, or more generally in an operational context where observers associated to $\mathcal{A}$ are restricted to operations in $\mathcal{A}$, the second term alone quantifies the distillable entanglement \cite{schuch2004nonlocal, soni2016aspects, acoleyen2015entanglement}.  For this reason, the entanglement entropy is sometimes defined using the second term alone.  

A third option for defining entanglement entropy in gauge theories is therefore to use equation \eqref{eq:EE_algebraic} but dropping the first term.

\bibliographystyle{plainnat}
\bibliography{KitaevModel_bib}

\begin{thebibliography}{27}
\providecommand{\natexlab}[1]{#1}
\providecommand{\url}[1]{\texttt{#1}}
\expandafter\ifx\csname urlstyle\endcsname\relax
  \providecommand{\doi}[1]{doi: #1}\else
  \providecommand{\doi}{doi: \begingroup \urlstyle{rm}\Url}\fi

\bibitem[Alicki et~al.(2007)Alicki, Fannes, and
  Horodecki]{alicki2007statistical}
R~Alicki, M~Fannes, and M~Horodecki.
\newblock A statistical mechanics view on {K}itaev's proposal for quantum
  memories.
\newblock \emph{Journal of Physics A: Mathematical and Theoretical},
  40\penalty0 (24):\penalty0 6451, 2007.
\newblock \doi{10.1088/1751-8113/40/24/012}.

\bibitem[Bachmann(2017)]{bachmann2017local}
Sven Bachmann.
\newblock Local disorder, topological ground state degeneracy and entanglement
  entropy, and discrete anyons.
\newblock \emph{Reviews in Mathematical Physics}, 29\penalty0 (06):\penalty0
  1750018, 2017.
\newblock \doi{10.1142/S0129055X17500180}.

\bibitem[Baez(1996)]{baez1996spin}
John~C Baez.
\newblock Spin networks in gauge theory.
\newblock \emph{Advances in Mathematics}, 117\penalty0 (2):\penalty0 253--272,
  1996.
\newblock \doi{10.1006/aima.1996.0012}.

\bibitem[Bravyi and Hastings(2011)]{bravyi2011short}
Sergey Bravyi and Matthew~B Hastings.
\newblock A short proof of stability of topological order under local
  perturbations.
\newblock \emph{Communications in mathematical physics}, 307\penalty0
  (3):\penalty0 609, 2011.
\newblock \doi{10.1007/s00220-011-1346-2}.

\bibitem[Bravyi et~al.(2010)Bravyi, Hastings, and
  Michalakis]{bravyi2010topological}
Sergey Bravyi, Matthew~B Hastings, and Spyridon Michalakis.
\newblock Topological quantum order: stability under local perturbations.
\newblock \emph{Journal of mathematical physics}, 51\penalty0 (9):\penalty0
  093512, 2010.
\newblock \doi{10.1063/1.3490195}.

\bibitem[Buerschaper and Aguado(2009)]{buerschaper2009mapping}
Oliver Buerschaper and Miguel Aguado.
\newblock Mapping {K}itaev's quantum double lattice models to {L}evin and
  {W}en's string-net models.
\newblock \emph{Physical Review B}, 80\penalty0 (15):\penalty0 155136, 2009.
\newblock \doi{10.1103/PhysRevB.80.155136}.

\bibitem[Buerschaper et~al.(2013)Buerschaper, Mombelli, Christandl, and
  Aguado]{buerschaper2013hierarchy}
Oliver Buerschaper, Juan~Mart{\'\i}n Mombelli, Matthias Christandl, and Miguel
  Aguado.
\newblock A hierarchy of topological tensor network states.
\newblock \emph{Journal of Mathematical Physics}, 54\penalty0 (1):\penalty0
  012201, 2013.
\newblock \doi{10.1063/1.4773316}.

\bibitem[Casini et~al.(2014)Casini, Huerta, and Rosabal]{casini2014remarks}
Horacio Casini, Marina Huerta, and Jos{\'e}~Alejandro Rosabal.
\newblock Remarks on entanglement entropy for gauge fields.
\newblock \emph{Physical Review D}, 89\penalty0 (8):\penalty0 085012, 2014.
\newblock \doi{10.1103/PhysRevD.89.085012}.

\bibitem[Cha et~al.(2018)Cha, Naaijkens, and Nachtergaele]{cha2018complete}
Matthew Cha, Pieter Naaijkens, and Bruno Nachtergaele.
\newblock The complete set of infinite volume ground states for {K}itaev's
  {A}belian quantum double models.
\newblock \emph{Communications in Mathematical Physics}, 357\penalty0
  (1):\penalty0 125--157, 2018.
\newblock \doi{10.1007/s00220-017-2989-4}.

\bibitem[Chang(2014)]{chang2014kitaev}
Liang Chang.
\newblock Kitaev models based on unitary quantum groupoids.
\newblock \emph{Journal of Mathematical Physics}, 55\penalty0 (4):\penalty0
  041703, 2014.
\newblock \doi{10.1063/1.4869326}.

\bibitem[Freedman et~al.(2002)Freedman, Larsen, and Wang]{freedman2002modular}
Michael~H Freedman, Michael Larsen, and Zhenghan Wang.
\newblock A modular functor which is universal for quantum computation.
\newblock \emph{Communications in Mathematical Physics}, 227\penalty0
  (3):\penalty0 605--622, 2002.
\newblock \doi{10.1007/s002200200645}.

\bibitem[Gaiotto et~al.(2015)Gaiotto, Kapustin, Seiberg, and
  Willett]{gaiotto2015generalized}
Davide Gaiotto, Anton Kapustin, Nathan Seiberg, and Brian Willett.
\newblock Generalized global symmetries.
\newblock \emph{Journal of High Energy Physics}, 2015\penalty0 (2):\penalty0
  172, 2015.
\newblock \doi{10.1007/JHEP02(2015)172}.

\bibitem[Kitaev(2003)]{kitaev2003fault}
A~Yu Kitaev.
\newblock Fault-tolerant quantum computation by anyons.
\newblock \emph{Annals of Physics}, 303\penalty0 (1):\penalty0 2--30, 2003.
\newblock \doi{10.1016/S0003-4916(02)00018-0}.

\bibitem[Kitaev and Preskill(2006)]{Kitaev:2005dm}
Alexei Kitaev and John Preskill.
\newblock {Topological entanglement entropy}.
\newblock \emph{Phys. Rev. Lett.}, 96:\penalty0 110404, 2006.
\newblock \doi{10.1103/PhysRevLett.96.110404}.

\bibitem[Levin and Wen(2006)]{levin2006detecting}
Michael Levin and Xiao-Gang Wen.
\newblock Detecting topological order in a ground state wave function.
\newblock \emph{Physical review letters}, 96\penalty0 (11):\penalty0 110405,
  2006.
\newblock \doi{10.1103/PhysRevLett.96.110405}.

\bibitem[Levin and Wen(2005)]{levin2005string}
Michael~A Levin and Xiao-Gang Wen.
\newblock String-net condensation: A physical mechanism for topological phases.
\newblock \emph{Physical Review B}, 71\penalty0 (4):\penalty0 045110, 2005.
\newblock \doi{10.1103/PhysRevB.71.045110}.

\bibitem[Lin and Radi{\v{c}}evi{\'c}(2020)]{lin2018comments}
Jennifer Lin and {\DJ}or{\dj}e Radi{\v{c}}evi{\'c}.
\newblock Comments on defining entanglement entropy.
\newblock \emph{Nuclear Physics B}, 958:\penalty0 115118, 2020.
\newblock \doi{10.1016/j.nuclphysb.2020.115118}.

\bibitem[Mazenc and Ranard(2019)]{mazenc2019target}
Edward~A Mazenc and Daniel Ranard.
\newblock Target space entanglement entropy.
\newblock \emph{arXiv preprint arXiv:1910.07449}, 2019.
\newblock URL \url{https://arxiv.org/abs/1910.07449}.

\bibitem[Naaijkens(2012)]{naaijkens2012anyons}
Pieter Naaijkens.
\newblock \emph{Anyons in infinite quantum systems: {QFT} in $d= 2+ 1$ and the
  toric code}.
\newblock PhD thesis, Radboud Universiteit Nijmegen, 2012.
\newblock URL \url{https://hdl.handle.net/2066/92737}.

\bibitem[Nielsen and Chuang(2011)]{nielsonchuang}
Michael~A. Nielsen and Isaac~L. Chuang.
\newblock \emph{Quantum Computation and Quantum Information: 10th Anniversary
  Edition}.
\newblock Cambridge University Press, New York, NY, USA, 10th edition, 2011.
\newblock ISBN 1107002176, 9781107002173.
\newblock \doi{10.1119/1.1463744}.

\bibitem[Ohya and Petz(2004)]{ohya2004quantum}
Masanori Ohya and D{\'e}nes Petz.
\newblock \emph{Quantum entropy and its use}.
\newblock Springer Science \& Business Media, 2004.
\newblock \doi{10.1016/0079-6727(95)90032-2}.

\bibitem[Schuch et~al.(2004)Schuch, Verstraete, and Cirac]{schuch2004nonlocal}
Norbert Schuch, Frank Verstraete, and J~Ignacio Cirac.
\newblock Nonlocal resources in the presence of superselection rules.
\newblock \emph{Physical review letters}, 92\penalty0 (8):\penalty0 087904,
  2004.
\newblock \doi{10.1103/PhysRevLett.92.087904}.

\bibitem[Sengupta(1994)]{sengupta1994gauge}
Ambar Sengupta.
\newblock Gauge invariant functions of connections.
\newblock \emph{Proceedings of the American Mathematical Society}, 121\penalty0
  (3):\penalty0 897--905, 1994.
\newblock \doi{10.1090/S0002-9939-1994-1215205-7}.

\bibitem[Soni and Trivedi(2016)]{soni2016aspects}
Ronak~M Soni and Sandip~P Trivedi.
\newblock Aspects of entanglement entropy for gauge theories.
\newblock \emph{Journal of High Energy Physics}, 2016\penalty0 (1):\penalty0
  136, 2016.
\newblock \doi{10.1007/JHEP01(2016)136}.

\bibitem[Van~Acoleyen et~al.(2016)Van~Acoleyen, Bultinck, Haegeman, Marien,
  Scholz, and Verstraete]{acoleyen2015entanglement}
Karel Van~Acoleyen, Nick Bultinck, Jutho Haegeman, Michael Marien, Volkher~B
  Scholz, and Frank Verstraete.
\newblock Entanglement of distillation for lattice gauge theories.
\newblock \emph{Physical Review Letters}, 117\penalty0 (13):\penalty0 131602,
  2016.
\newblock \doi{10.1103/PhysRevLett.117.131602}.

\bibitem[Wall(1947)]{wall1947finite}
GE~Wall.
\newblock Finite groups with class-preserving outer automorphisms.
\newblock \emph{Journal of the London Mathematical Society}, 1\penalty0
  (4):\penalty0 315--320, 1947.
\newblock \doi{10.1112/jlms/s1-22.4.315}.

\bibitem[Wong(2018)]{wong2018note}
Gabriel Wong.
\newblock A note on entanglement edge modes in {C}hern {S}imons theory.
\newblock \emph{Journal of High Energy Physics}, 2018\penalty0 (8):\penalty0
  20, 2018.
\newblock \doi{10.1007/JHEP08(2018)020}.

\end{thebibliography}

\end{document}